\documentclass[11pt]{article}

\usepackage{morefloats}
    \usepackage{amsfonts}
\usepackage{color}
\usepackage{multirow}
\usepackage{latexsym}
\usepackage{amsmath,amssymb,amsthm}
\usepackage{dsfont}
\usepackage{extarrows}
    \newcommand{\N}{\mathbb{N}}
    \newcommand{\E}{\mathbb{E}}

    \newcommand{\R}{\mathbb{R}}
    
    \newcommand{\norm}[1]{\left\lVert #1 \right\rVert}

\newtheorem{theorem}{Theorem}[section]
  \newtheorem{lemma}[theorem]{Lemma}
  \newtheorem{corollary}[theorem]{Corollary}
  \newtheorem{proposition}[theorem]{Proposition}
  \theoremstyle{definition}
  
  \newtheorem*{remark}{Remark}
	\renewenvironment{proof}{\textbf{\emph{Proof.}}}{\qed}
	\renewenvironment{remark}{\textbf{\emph{Remark.}}}

\makeatletter
\def\blfootnote{\xdef\@thefnmark{}\@footnotetext}
\makeatother

\begin{document}

\title{\bf Testing normality via a distributional fixed point property in the Stein characterization }

%\titlerunning{Testing normality via a Steinian fixed point property}        % if too long for running head

\author{S. Betsch and B. Ebner }

\date{\today}
\maketitle

\blfootnote{ {\em American Mathematical Society 2010 subject
classifications.} Primary 62G10 Secondary 60F05} 
\blfootnote{
{\em Key words and phrases} Goodness-of-fit, Normal Distribution, Stein's Method, Zero Bias Transformation}

\begin{abstract} We propose two families of tests for the classical goodness-of-fit problem to univariate normality. The new procedures are based on $L^2$-distances of the empirical zero-bias transformation to the normal distribution or the empirical distribution of the data, respectively. Weak convergence results are derived under the null hypothesis, under fixed alternatives as well as under contiguous alternatives. Empirical critical values are provided and a comparative finite-sample power study shows the competitiveness to classical procedures.
\end{abstract}

\section{Introduction}
\label{Intro}
Testing normality is commonly known as the most used and discussed goodness-of-fit technique, motivated by the model assumption of normality in classical models. To be specific, let $X, X_1, X_2, \dotso$ be real valued independent and identically distributed (iid.) random variables. The problem of interest is to test the hypothesis
\begin{equation}\label{H0}
	H_0: \mathbb{P}^X \in \mathcal{N}=\{ \mathcal{N}(\mu, \sigma^2) \, | \, (\mu, \sigma^2) \in \R \times (0,\infty) \}
\end{equation}
against general alternatives. So far, a great variety of goodness-of-fit tests have been proposed and research is of ongoing interest, as witnessed by the recent papers \cite{BGWX:2016,HJG:2018,VG:2015} and comparative simulation studies like \cite{RDC:2010,YS:2011}. Classical procedures in goodness-of-fit methodology as the Kolmogorov-Smirnov and the Cram\'er-von Mises test approach the testing problem by measuring the distance of the empirical distribution function to the estimated representant of $\mathcal{N}$. For a theoretical approach to the goodness-of-fit to a family of distributions, see \cite{delBarrio2000,N:1979}. Other methods are based on skewness and kurtosis, as for instance proposed in \cite{PDB:1977} (known to lead to inconsistent procedures), the empirical characteristic function, see \cite{EP:1983}, the Wasserstein distance, see \cite{delBarrio2000,DBCMRR:1999}, the sample entropy, see \cite{V:1976}, or correlation and regression tests, as the famous Shapiro-Wilk test, see \cite{SW:1965}, among others. For a survey of classical existing methods see \cite{delBarrio2000}, section 3, and \cite{H:1994}, for surveys on goodness-of-fit techniques connected to characterizations of distributions, including the normal, see \cite{MS:2013a,MS:2013b,N:2017} and for the problem of testing multivariate normality, we refer to \cite{H:2002,MM:2004}.

Another natural approach to assess the distance of the distribution of a real valued random variable $X$ to the normal distribution is to calculate the difference between $\E h(X)$ and $\E h(N)$, where $N \sim \mathcal{N}(0,1)$, over some large class of functions $h : \R \to \R$. With the class $\{e^{itx}\, | \, t \in \R \}$ leading to the characteristic functions of the distributions one heavily relies on the assumption of independence when proving limit theorems. In an attempt to give an alternative proof of the central limit theorem, Charles Stein considered a different class of test functions (see eg. \cite{Stein:1972}). Stating that $X$ has a standard normal distribution if, and only if,
\begin{equation} \label{Stein characterization}
	\E \big[ f^{\prime}(X) - X f(X) \big] = 0
\end{equation}
holds for each absolutely continuous function $f$ for which the expectation exists, it appears reasonable to regard the left hand side of ($\ref{Stein characterization}$), for a suitable function $f$, as an estimate of $\E h(X) - \E h(N)$ since both terms ought to be small whenever the distribution of $X$ is close to standard normal. In practice, solving the differential equation
\begin{equation} \label{diff equation}
	f^{\prime}(x) - x f(x) = h(x) - \E h(N)
\end{equation}
for absolutely continuous functions $h$, evaluating at $X$ and taking expectations, the problem reduces to appraising $\E [ f_h^{\prime}(X) - X f_h(X)]$, with $f_h$ being the solution of ($\ref{diff equation}$). A commonly used tool to handle these terms is the so called zero bias transformation introduced by \cite{GR:1997}. Namely, if $\E X = 0$ and $\mathbb{V}(X) = 1$, a random variable $X^*$ is said to have the $X$-zero bias distribution if
\begin{equation} \label{Def Zero Bias}
	\E \big[ f^{\prime}(X^*)\big] = \E \big[ X f(X) \big]
\end{equation}
holds for all absolutely continuous functions $f$ for which these expectations exist. The use of this distribution, if existent, lends itself easily to the purpose of distributional approximation. For instance, starting with the solution of (\ref{diff equation}), the mean value theorem gives
\begin{equation*}
	| \E h(X) - \E h(N) |
	= | \E [f_h^{\prime}(X) - f_h^{\prime}(X^*)] |
	\leq \norm{f_h^{\prime \prime}}_{\infty} \, \E | X - X^* | .
\end{equation*}
Thus, the problem reduces to bounding the derivatives of the solution $f_h$ of ($\ref{diff equation}$) and constructing $X^*$ such that $\E | X - X^* |$ is accessible. Bounds on $f_h$ and its derivatives are well-known and a comprehensive treatment as well as explicit constructions for $X^*$ may be found in \cite{CGS:2011} (for the bounds see also \cite{Stein:1986}). For an introduction to Stein's method, see \cite{CGS:2011,R:2011}. One of the main reasons Stein's method, particularly for the normal distribution, has been studied to a remarkable extent are various central limit type results, also giving convergence rates, even in dependency settings.

It seems reasonable to ask whether Stein's characterization ($\ref{Stein characterization}$) may be used to construct a goodness-of-fit statistic. Apparently, we can hardly evaluate a quantity for all absolutely continuous functions which makes the direct application of equation ($\ref{Stein characterization}$) rather complicated (cf. \cite{LLJ:2016}). Instead, we propose a test based on the zero bias distribution. To this end, we interpret ($\ref{Def Zero Bias}$) as a distributional transformation $\mathbb{P}^X \mapsto \mathbb{P}^{X^*}$ and notice that, by ($\ref{Stein characterization}$), the standard normal distribution is the unique fixed point of this transformation, see \cite{GR:1997}, Lemma 2.1 (i) or \cite{CGS:2011}.
\begin{proposition} \label{Prop existence zero bias}
	If $X$ is a centred, real valued random variable with $\mathbb{V}(X) = 1$, the $X$-zero bias distribution exists and is unique. Moreover, it is absolutely continuous with respect to the Lebesgue measure with density
	\begin{equation*}
		d^X (t) = \E [X \mathds{1}\{ X > t \}] = - \E [X \, \mathds{1}\{ X \leq t \}]
	\end{equation*}
	and distribution function
	\begin{equation*}
		F^X (t) = \E [X (X - t) \mathds{1}\{ X \leq t \}] .
	\end{equation*}	
\end{proposition}
A proof can be found in \cite{CGS:2011} or in the original treatment \cite{GR:1997}. In view of Proposition $\ref{Prop existence zero bias}$ and the interpretation of ($\ref{Def Zero Bias}$) as a distributional transformation, the normal distribution is characterized as follows. A random variable $X$ with distribution function $F$ and $\E X = 0$, $\mathbb{V}(X) = 1$ is standard normally distributed if, and only if,
\begin{equation*}
	F^X = F = \Phi ,
\end{equation*}
where $\Phi$ is the distribution function of the standard normal distribution. In section $\ref{distribution function tests}$ we will use a weighted $L^2$-measure of deviation between an empirical version of $F^X$ and $\Phi$ or the empirical distribution, respectively, to construct two statistics for our testing problem (\ref{H0}). Further, we will establish the consistency of our classes of tests and derive their limit distributions under fixed alternatives in section $\ref{Section under fixed alternatives}$. The limit null distributions are derived in section $\ref{limit null}$. The behaviour under contiguous alternatives is studied in section $\ref{contiguous alternatives}$ and empirical results in form of a power study are presented in section $\ref{Section emp results}$. Conclusions and outlines complete the article.

\section{The proposed test statistics}
\label{distribution function tests}
Let $X, X_1, X_2, \dotso$ be real valued iid. random variables defined on an underlying probability space $(\Omega, \mathcal{A}, \mathbb{P})$. Further, let $F$ be the distribution function of $X$ and assume that $\E [X^2] < \infty$. We intend to test the hypothesis in (\ref{H0}). To reflect the invariance of the family of normal distrbutions $\mathcal{N}$ with respect to affine transformations, the proposed statistics will only depend on the so called scaled residuals, namely $Y_{n,1}, \dots, Y_{n,n}$, $Y_{n,j} = (X_j - \overline{X}_n)/S_n$, where $\overline{X}_n = n^{-1} \sum_{k=1}^{n} X_k$ and $S_n^2 = n^{-1} \sum_{k=1}^{n}(X_k - \overline{X}_n)^2$ are the sample mean and variance, respectively. This way, the values of our statistics themselves and thus the tests based on them are invariant under affine transformation of the data. We note that $X$ has a normal distribution with some parameters $\mu$ and $\sigma^2$ if, and only if, $Y_{n,1}$ is asymptotically standard normal as $(\overline{X}_n, S_n^2)$ is a strongly consistent estimator of $(\mu, \sigma^2)$. Due to the affine invariance, we assume w.l.o.g. $\E X = 0$ and $\mathbb{V}(X) = 1$.

In view of ($\ref{Stein characterization}$), ($\ref{Def Zero Bias}$) and Proposition $\ref{Prop existence zero bias}$, we suggest the Cram\'er-von Mises type test statistics
\begin{equation} \label{teststat 1}
	G_n^{(1)} = n \int_{\R} \left( \frac{1}{n} \sum_{j=1}^{n} Y_{n,j} (Y_{n,j} - t) \mathds{1}\{ Y_{n,j} \leq t \} - \frac{1}{n} \sum_{j=1}^{n} \mathds{1}\{ Y_{n,j} \leq t \} \right)^2 \omega(t) \, \mathrm{d}t
\end{equation}
and
\begin{equation} \label{teststat 2}
	G_n^{(2)} = n \int_{\R} \left( \frac{1}{n} \sum_{j=1}^{n} Y_{n,j} (Y_{n,j} - t) \mathds{1}\{ Y_{n,j} \leq t \} - \Phi(t) \right)^2 \omega(t) \, \mathrm{d}t .
\end{equation}
Here, $n^{-1} \sum_{j=1}^{n} Y_{n,j} (Y_{n,j} - t) \mathds{1}\{ Y_{n,j} \leq t \}$ is an empirical version of the zero bias distribution function and $n^{-1} \sum_{j=1}^{n} \mathds{1}\{ Y_{n,j} \leq t \}$ is the empirical distribution function of $Y_{n,1}, \dots, Y_{n,n}$. By $\omega : \R \to \R$ we denote a positive continuous weight function satisfying
\begin{equation} \label{requirements weight function}
	\int_{\R} t^4 \, \omega(t) \, \mathrm{d}t < \infty .
\end{equation}
A test based on $G_n^{(1)}$ or $G_n^{(2)}$ rejects $H_0$ for large values of the statistic. A suggestion for the choice of the weight function and equivalent expressions for $G_n^{(1)}$ and $G_n^{(2)}$ that are suitable for computations can be found in section $\ref{Section emp results}$. \\
For the asymptotic theory we consider the Hilbert space $L^2 = L^2(\R, \mathcal{B}, \omega \, \mathrm{d} \mathcal{L}^1)$ of measurable, square integrable functions $f: \R \to \R$. We denote by
\begin{equation*}
	\norm{f}_{L^2} = \left( \int_{\R} f(t)^2 \, \omega(t) \, \mathrm{d}t \right)^{1/2},\qquad \langle f, g \rangle_{L^2}=\int_{\R} f(t)g(t) \, \omega(t) \, \mathrm{d}t
\end{equation*}
the usual $L^2$-norm as well as the usual inner product in $L^2$. Notice that the functions figuring within the norm in the definition of $G_n^{(1)}$ and $G_n^{(2)}$ are random elements of $L^2$. In the following, we denote convergence in distribution by $\stackrel{\mathcal{D}}{\longrightarrow}$ and write $o_{\mathbb{P}}(1)$ and $O_{\mathbb{P}}(1)$ for convergence to $0$ in probability and boundedness in probability, respectively. Before we present our main results for the statistics, we proof a preliminary Lemma which is frequently used in the subsequent elaborations.
\begin{lemma} \label{prelim lemma}
	For $X, X_1, X_2, \dotso$ as above, we set
	\begin{equation*}
		\widehat{F}_n^X (s) = \frac{1}{n} \sum_{j=1}^{n} \frac{X_j - \overline{X}_n}{S_n^2} \, (X_j - s) \, \mathds{1}\{X_j \leq s\},\quad s\in\mathbb{R} .
	\end{equation*}
	Then
	\begin{equation} \label{Gliv Cant for emp zero bias}
		\sup\limits_{s \, \in \, \R} \left| \widehat{F}_n^X (s) - F^X (s) \right| \longrightarrow 0
	\end{equation}
	$\mathbb{P}$-almost surely. Additionally, putting
	\begin{equation*}
		A_n(s) = \frac{1}{n} \sum_{j=1}^{n} (X_j - s) \, \mathds{1}\{X_j \leq s\} - \E \big[(X - s) \, \mathds{1}\{X \leq s\} \big] ,\quad s\in\mathbb{R},
	\end{equation*}
	we have
	\begin{equation} \label{unif almost sure conv}
		\int_{\R} \big( A_n(s) \big)^2 \omega \left(\frac{s - \overline{X}_n}{S_n}\right) \, \mathrm{d}s
		= o_{\mathbb{P}}(1) .
	\end{equation}
\end{lemma}
\begin{proof}
	First, notice that
	\begin{equation*}
		\widehat{d}_n^X (s)
		= \frac{1}{n} \sum_{j=1}^{n} \frac{X_j - \overline{X}_n}{S_n^2} \, \mathds{1}\{X_j > s\}
		= - \frac{1}{n} \sum_{j=1}^{n} \frac{X_j - \overline{X}_n}{S_n^2} \, \mathds{1}\{X_j \leq s\}
		~~(\geq 0) .
	\end{equation*}
	Using the first representation when integrating over $(\overline{X}_n, \infty)$ and the second for $(- \infty, \overline{X}_n)$, we obtain
	\begin{equation*}
		\int_{\R} \widehat{d}_n^X (t) \, \mathrm{d}t
		= \frac{1}{S_n^2} \left( \frac{1}{n} \sum_{j=1}^{n} \big( X_j - \overline{X}_n \big)^2 \right)
		= 1.
	\end{equation*}
	Now, since
	\begin{equation*}
		\int_{- \infty}^{s} \widehat{d}_n^X (t) \, \mathrm{d}t
		= - \frac{1}{n} \sum_{j=1}^{n} \frac{X_j - \overline{X}_n}{S_n^2} \int_{- \infty}^{s} \mathds{1}\{X_j \leq t\} \, \mathrm{d}t
		= \widehat{F}_n^X (s) ,
	\end{equation*}
	we conclude that $\widehat{F}_n^X$ is a continuous distribution function. We observe that, by the strong law of large numbers and the almost sure convergence $(\overline{X}_n, S_n^2) \to (0,1)$, we have
	\begin{align*}
		\widehat{F}_n^X (s)
		&= \frac{1}{S_n^2} \cdot \frac{1}{n} \sum_{j=1}^{n} X_j (X_j - s) \mathds{1}\{ X_j \leq s \} - \frac{\overline{X}_n}{S_n^2} \cdot \frac{1}{n} \sum_{j=1}^{n} (X_j - s) \mathds{1}\{ X_j \leq s \} \\
		&\longrightarrow F^X (s)
	\end{align*}
	$\mathbb{P}$-almost surely for $n \to \infty$ and any $s \in \R$. The proof of the classical Glivenko-Cantelli theorem applies to $\widehat{F}_n^X$ which yields ($\ref{Gliv Cant for emp zero bias}$). For ($\ref{unif almost sure conv}$), we record that straightforward calculations, using Fubini's theorem and ($\ref{requirements weight function}$), give
	\begin{equation*}
		\E \left[ \int_{\R} A_n(s)^2 \, \omega(s) \, \mathrm{d}s \right] \longrightarrow 0 ~~ \text{as} ~ n \to \infty ,
	\end{equation*}
	so $\norm{A_n}_{L^2}^2 = o_{\mathbb{P}}(1)$. Noting that $\omega$ is continuous, we have
	\begin{equation} \label{unif conv weight function}
		\sup\limits_{s \, \in \, \R} \left| \omega\left( \frac{s - \overline{X}_n}{S_n} \right) \bigg/ \omega(s) - 1 \right| \longrightarrow 0 ~~ \text{as} ~ n \to \infty
	\end{equation}
	on a set of measure one. Therefore,
	\begin{align*}
		&\left| \int_{\R} \big(A_n(s)\big)^2 \, \omega\left(\frac{s - \overline{X}_n}{S_n}\right) \mathrm{d}s - \int_{\R} \big(A_n(s)\big)^2 \, \omega(s) \, \mathrm{d}s \right| \\
		&~~~~~~~~~~~~~ \leq \sup\limits_{s \, \in \, \R} \left| \omega\left( \frac{s - \overline{X}_n}{S_n} \right) \bigg/ \omega(s) - 1 \right| \cdot \norm{A_n}_{L^2}^2 \\
		&~~~~~~~~~~~~~ \longrightarrow 0
	\end{align*}
	in probability, for $n \to \infty$, which finishes the proof.
\end{proof}

\section{Consistency and limit distributions under fixed alternatives}
\label{Section under fixed alternatives}
A first use of Lemma $\ref{prelim lemma}$ becomes evident when realizing that, by a simple change of variable,
\begin{equation} \label{teststat 1 after change of var}
	G_n^{(1)} = \frac{n}{S_n} \int_{\R} \left( \widehat{F}_n^X (s) - \frac{1}{n} \sum_{j=1}^{n} \mathds{1}\{ X_j \leq s \} \right)^2 \omega\left( \frac{s - \overline{X}_n}{S_n} \right) \, \mathrm{d}s
\end{equation}
and
\begin{equation} \label{teststat 2 after change of var}
	G_n^{(2)} = \frac{n}{S_n} \int_{\R} \left( \widehat{F}_n^X (s) - \Phi\left( \frac{s - \overline{X}_n}{S_n} \right) \right)^2 \omega\left( \frac{s - \overline{X}_n}{S_n} \right) \, \mathrm{d}s .
\end{equation}
In the broad, non-parametric setting introduced at the beginning of section \ref{distribution function tests}, the limit in (\ref{Gliv Cant for emp zero bias}), the classical theorem of Glivenko-Cantelli and ($\ref{unif conv weight function}$) yield the following theorem.
\begin{theorem} \label{thm consistency}
	As $n \to \infty$, we have
	\begin{equation*}
		\frac{G_n^{(1)}}{n}
		\longrightarrow \int_{\R} \left( F^X (s) - F(s) \right)^2 \omega(s) \, \mathrm{d}s
		= \norm{F^X - F}_{L^2}^2
		= \Delta^{(1)}
	\end{equation*}
	and
	\begin{equation*}
		\frac{G_n^{(2)}}{n}
		\longrightarrow \int_{\R} \left( F^X (s) - \Phi(s) \right)^2 \omega(s) \, \mathrm{d}s
		= \norm{F^X - \Phi}_{L^2}^2
		= \Delta^{(2)} ,
	\end{equation*}
	each convergence being $\mathbb{P}$-almost surely.
\end{theorem}

By Proposition $\ref{Prop existence zero bias}$ together with ($\ref{Stein characterization}$), the limits $\Delta^{(1)}$ and $\Delta^{(2)}$ are positive if $X$ has a non-normal distribution. Consequently, any level-$\alpha$-test based on $G_n^{(1)}$ or $G_n^{(2)}$ is consistent against each alternative with existing second moment. \\
We will now use ($\ref{unif almost sure conv}$) to specify Theorem $\ref{thm consistency}$. More precisely, we study the limit distributions of our statistics under fixed alternatives. In the following, we write $\mathcal{U}_n(s) \approx \mathcal{V}_n(s)$ whenever
\begin{equation*}
	\int_{\R} \big( \mathcal{U}_n(s) - \mathcal{V}_n(s) \big)^2 \, \omega\left( \frac{s - \overline{X}_n}{S_n} \right) \mathrm{d}s
	= o_{\mathbb{P}}(1) ~~ \text{as} ~ n \to \infty .
\end{equation*}
Here, $\mathcal{U}_n$ and $\mathcal{V}_n$ are random elements of our Hilbert space. By ($\ref{unif conv weight function}$), this new notation is equivalent to $\norm{\mathcal{U}_n - \mathcal{V}_n}_{L^2}^2 = o_{\mathbb{P}}(1)$. %Starting with $G_n^{(1)}$, we have
\begin{theorem} \label{thm teststat 1 conv under altern}
	Let $X, X_1, X_2, \dotso$ be iid., non-normal random variables with $\E X^4 < \infty$, distribution function $F$, continuously differentiable density $p$ satisfying $\sup_{s \, \in \, \R} |s \cdot p(s)| \leq K_1 \in (0, \infty)$ and $\sup_{s \, \in \, \R} |p^{\prime}(s)| \leq K_2 \in (0, \infty)$ and, w.l.o.g., $\E X = 0$, $\mathbb{V}(X) = 1$. Then, as $n\rightarrow\infty$,
	\begin{equation} \label{teststat 1 conv under altern}
		\sqrt{n} \left( \frac{G_n^{(1)}}{n} - \Delta^{(1)} \right)
		\stackrel{\mathcal{D}}{\longrightarrow}
		\mathcal{N}\left(0, \tau_{(1)}^2\right) ,
	\end{equation}
	where
	\begin{equation*}
		\tau_{(1)}^2 = 4 \int_{\R} \int_{\R} \mathcal{C}^{(1)}(s,t) \big( F^X(s) - F(s) \big) \big( F^X(t) - F(t) \big) \, \omega(s) \, \omega(t) \, \mathrm{d}s \, \mathrm{d}t
	\end{equation*}
	with
	\begin{align*}
		\mathcal{C}^{(1)}(s,t) = \E \big[ C^{(1)}(s) \, C^{(1)}(t) \big],\quad s,t\in\mathbb{R},
	\end{align*}
	and
	\begin{align*}
		C^{(1)}(s) = & \, X (X - s) \mathds{1}\{ X \leq s \} - X \, \E \big[ (X - s) \mathds{1}\{ X \leq s \} \big] - X^2 \, F^X(s) \\
		&- \mathds{1}\{ X \leq s \} + F(s) - \left( \tfrac{1}{2} (1 - X^2) \cdot s - X \right) \big( d^X(s) - p(s) \big),\quad s\in\mathbb{R} .
	\end{align*}
\end{theorem}
\begin{proof}
	The main idea of the proof is to approximate the term figuring within the integral in ($\ref{teststat 1 after change of var}$) by ($n^{-1/2}$ times) a sum of iid. random elements of $L^2$, for which the central limit theorem in Hilbert spaces is applicable. \\
	Setting
	\begin{equation*}
		U_n (s) = \widehat{F}_n^X (s) - \frac{1}{n} \sum\limits_{j=1}^{n} \mathds{1}\{ X_j \leq s \} - F^X \left( \frac{s - \overline{X}_n}{S_n} \right) + F \left( \frac{s - \overline{X}_n}{S_n} \right) ,\quad s\in\mathbb{R},
	\end{equation*}
	a change of variable in both integrals and an integral decomposition, as used by \cite{Chap:1958}, gives
	\begin{align} \label{decomp Chapman}
		\sqrt{n} \left( \frac{G_n^{(1)}}{n} - \Delta^{(1)} \right)
		= & ~ \frac{1}{S_n} \left\{ 2 \int_{\R} \sqrt{n} \, U_n(s) \cdot \left[ F^X (\widetilde{s}) - F (\widetilde{s}) \right] \omega (\widetilde{s}) \, \mathrm{d}s \right. \nonumber \\
		&\left. ~~~~~~ + \frac{1}{\sqrt{n}} \int_{\R} \big( \sqrt{n} \, U_n(s) \big)^2 \, \omega (\widetilde{s}) \, \mathrm{d}s \right\} ,
	\end{align}
	where $\widetilde{s}=S_n^{-1} (s - \overline{X}_n)$. Invoking Proposition $\ref{Prop existence zero bias}$, a Taylor expansion yields
	\begin{equation} \label{taylor exp F^X}
		F^X \left( \frac{s - \overline{X}_n}{S_n} \right)
		= F^X(s) + d^X(s) \left( s \cdot \left( \frac{1}{S_n} - 1 \right) - \frac{\overline{X}_n}{S_n} \right)
		- R_n(s) ,
	\end{equation}
	where
	\begin{equation*}
		R_n(s)
		= \frac{1}{2} \, \xi_n \, p(\xi_n) \left( s \cdot \left( \frac{1}{S_n} - 1 \right) - \frac{\overline{X}_n}{S_n} \right)^2
	\end{equation*}
	with $| \xi_n - s | \leq | s \cdot (S_n^{-1} - 1) - S_n^{-1} \, \overline{X}_n |$. Using $\sup_{s \, \in \, \R} |s \cdot p(s)| \leq K_1$, $\sqrt{n} (S_n^{-1} - 1) = O_{\mathbb{P}}(1)$ and $\sqrt{n} \, \overline{X}_n = O_{\mathbb{P}}(1)$, we have $\sqrt{n} \, R_n(s) \approx 0$ (observe that we need ($\ref{requirements weight function}$) to guarantee that $R_n \in L^2$). With an analogous expansion for $F$, utilizing $\sup_{s \, \in \, \R} |p^{\prime}(s)| \leq K_2$, we note
	\begin{equation} \label{taylor exp F}
		\sqrt{n} \, F \left( \frac{s - \overline{X}_n}{S_n} \right)
		\approx \sqrt{n} \left\{ F(s) + p(s) \left( s \cdot \left( \frac{1}{S_n} - 1 \right) - \frac{\overline{X}_n}{S_n} \right) \right\} .
	\end{equation}
	Now, by ($\ref{unif almost sure conv}$) of Lemma $\ref{prelim lemma}$ and $\sqrt{n} \, \overline{X}_n = O_{\mathbb{P}}(1)$,
	\begin{align} \label{approx of F_n^X}
		\sqrt{n} \, \widehat{F}_n^X (s)
		\approx & ~ \frac{\sqrt{n}}{S_n^2} \left\{ \frac{1}{n} \sum\limits_{j=1}^{n} X_j (X_j - s) \mathds{1}\{ X_j \leq s \} \right. \nonumber \\
		&\left. ~~~~~~~~ - \overline{X}_n \, \E \big[ (X - s) \mathds{1}\{X \leq s\} \big] \vphantom{\sum\limits_{j=1}^{n}} \right\} .
	\end{align}
	Combining ($\ref{approx of F_n^X}$) and the expansions ($\ref{taylor exp F^X}$), ($\ref{taylor exp F}$) we get
	\begin{align*}
		\sqrt{n} \, U_n(s)
		\approx & \, \frac{\sqrt{n}}{S_n^2} \left\{ \frac{1}{n} \sum\limits_{j=1}^{n} X_j (X_j - s) \mathds{1}\{ X_j \leq s \} - \overline{X}_n \, \E \big[ (X - s) \mathds{1}\{X \leq s\} \big] \right. \\
		&\left. ~~~~~~~ - S_n^2 \left( \frac{1}{n} \sum\limits_{j=1}^{n} \mathds{1}\{ X_j \leq s \} - F(s) \right) - S_n^2 \, F^X(s) \right. \\
		&\left. ~~~~~~~ - S_n \big( d^X(s) - p(s) \big) \big( s\cdot (1 - S_n) - \overline{X}_n \big) \vphantom{\sum\limits_{j}^{n}}\right\} .
	\end{align*}
	Again, using $\sqrt{n} \, \overline{X}_n = O_{\mathbb{P}}(1)$, we realize that
	\begin{equation*}
		\sqrt{n} \, S_n^2 \, F^X(s)
		= \left( \frac{1}{\sqrt{n}} \sum\limits_{j=1}^{n} X_j^2 - \sqrt{n} \, \overline{X}_n^2 \right) F^X(s)
		\approx \frac{1}{\sqrt{n}} \sum\limits_{j=1}^{n} X_j^2 \, F^X(s) .
	\end{equation*}
	Moreover, by the classical Glivenko-Cantelli theorem and $\sqrt{n} \, (S_n^2 - 1) = O_{\mathbb{P}}(1)$,
	\begin{align*}
		\sqrt{n} \, S_n^2 \left( \frac{1}{n} \sum\limits_{j=1}^{n} \mathds{1}\{ X_j \leq s \} - F(s) \right)
		\approx \frac{1}{\sqrt{n}} \sum\limits_{j=1}^{n} \big( \mathds{1}\{ X_j \leq s \} - F(s) \big) .
	\end{align*}
	Further, a Taylor expansion of the square root gives
	\begin{equation*}
		\sqrt{n} \, (1 - S_n)
		= - \sqrt{n} \, \left( \tfrac{1}{2} (S_n^2 - 1) - \tfrac{1}{8} \xi_n^{- 3/2} (S_n^2 - 1)^2 \right)
	\end{equation*}
	where $|\xi_n - 1| \leq |S_n^2 - 1|$. Therefore, subliminally assuming that $n$ is large enough to ensure $|S_n^2 - 1| < 1/2$ (on a set of measure one), $|\sqrt{n} \, \xi_n^{- 3/2} (S_n^2 - 1)^2| < 3 \sqrt{n} \, | (S_n^2 - 1)^2 | = o_{\mathbb{P}}(1)$ and
	\begin{equation*}
		\sqrt{n} \, (1 - S_n)
		\approx \frac{\sqrt{n}}{2} \big(1 - S_n^2\big)
		\approx \frac{1}{\sqrt{n}} \sum\limits_{j=1}^{n} \frac{1}{2} \big(1 - X_j^2\big) .
	\end{equation*}
	Putting these last four approximations together, we have
	\begin{equation} \label{asymp repr U_n}
		\sqrt{n} \, U_n(s) \approx \frac{1}{S_n^2} \, \frac{1}{\sqrt{n}} \sum\limits_{j=1}^{n} W_j(s) ,
	\end{equation}
	where
	\begin{align*}
		W_j(s) = & \, X_j (X_j - s) \mathds{1}\{ X_j \leq s \} - X_j \, \E \big[ (X - s) \mathds{1}\{ X \leq s \} \big] - X_j^2 \, F^X(s) \\
		&- \mathds{1}\{ X_j \leq s \} + F(s) - \left( \tfrac{1}{2} (1 - X_j^2) \cdot s - X_j \right) \big( d^X(s) - p(s) \big) .
	\end{align*}
	Notice that $W_1, \dots, W_n$ are iid. random elements of $L^2$ with $\E W_1 = 0$. Hence, the central limit theorem for separable Hilbert spaces (see \cite{LT:1991}, Corollary 10.9) implies
	\begin{equation*}
		\frac{1}{\sqrt{n}} \sum\limits_{j=1}^{n} W_j(\cdot) \stackrel{\mathcal{D}}{\longrightarrow} \mathcal{W}(\cdot) ,
	\end{equation*}
	where $\mathcal{W} \in L^2$ is a centred Gaussian element satisfying $\E \norm{\mathcal{W}}_{L^2}^2 < \infty$. With ($\ref{asymp repr U_n}$), the stochastic boundedness of $n^{- 1/2} \sum_{j=1}^{n} W_j$ and ($\ref{unif conv weight function}$), decomposition ($\ref{decomp Chapman}$) reads as
	\begin{equation*}
		\sqrt{n} \left( \frac{G_n^{(1)}}{n} - \Delta^{(1)} \right)
		= 2 \, \langle \sqrt{n} \, U_n , F^X - F \rangle_{L^2}
		+ \frac{1}{\sqrt{n}} \norm{\sqrt{n} \, U_n}_{L^2}^2
		+ o_{\mathbb{P}}(1) .
	\end{equation*}
	The continuous mapping theorem and Slutzki's Lemma imply
	\begin{equation*}
		\sqrt{n} \left( \frac{G_n^{(1)}}{n} - \Delta^{(1)} \right)
		\stackrel{\mathcal{D}}{\longrightarrow}
		2 \, \langle \mathcal{W}, F^X - F \rangle_{L^2} .
	\end{equation*}
	Letting
	\begin{equation*}
		\mathcal{C}^{(1)}(s,t) = \E \big[ \mathcal{W}(s) \, \mathcal{W}(t) \big]
	\end{equation*}
	be the covariance kernel of $\mathcal{W}$, we finally record that $2 \, \langle \mathcal{W}, F^X - F \rangle_{L^2}$ has the normal distribution $\mathcal{N}(0, \tau_{(1)}^2)$, where
	\begin{align*}
		\tau_{(1)}^2
		&= 4\, \E \big[ \langle \mathcal{W}, F^X - F \rangle_{L^2}^2 \big] \\
		&= 4 \int_{\R} \int_{\R} \mathcal{C}^{(1)}(s,t) \big( F^X(s) - F(s) \big) \big( F^X(t) - F(t) \big) \, \omega(s) \, \omega(t) \, \mathrm{d}s \, \mathrm{d}t ,
	\end{align*}
	ending the proof.
\end{proof}

Applying the reasoning of Theorem $\ref{thm teststat 1 conv under altern}$ to $G_n^{(2)}$, we realize that $\Phi$ will drop out when considering the decomposition of the integrals. Proceeding with the remaining terms exactly as before, we obtain an analogous statement for the second statistic under slightly weaker conditions.
\begin{corollary} \label{coro teststat 2 conv under altern}
	Let $X, X_1, X_2, \dotso$ be iid., non-normal random variables with $\E X^4 < \infty$, distribution function $F$ and Lebesgue density $p$. Further, assume $\sup_{s \, \in \, \R} |s \cdot p(s)| < \infty$ and, w.l.o.g, $\E X = 0$, $\mathbb{V}(X) = 1$. Then, as $n\rightarrow\infty,$
	\begin{equation} \label{teststat 2 conv under altern}
		\sqrt{n} \left( \frac{G_n^{(2)}}{n} - \Delta^{(2)} \right)
		\stackrel{\mathcal{D}}{\longrightarrow}
		\mathcal{N}\left(0, \tau_{(2)}^2\right) ,
	\end{equation}
	where
	\begin{equation*}
		\tau_{(2)}^2 = 4 \int_{\R} \int_{\R} \mathcal{C}^{(2)}(s,t) \big( F^X(s) - \Phi(s) \big) \big( F^X(t) - \Phi(t) \big) \, \omega(s) \, \omega(t) \, \mathrm{d}s \, \mathrm{d}t
	\end{equation*}
	with
	\begin{align*}
		\mathcal{C}^{(2)}(s,t) = \E \big[ C^{(2)}(s) \, C^{(2)}(t) \big],\quad s,t\in\mathbb{R},
	\end{align*}
	and
	\begin{align*}
		C^{(2)}(s) = & \, X (X - s) \mathds{1}\{ X \leq s \} - X \, \E \big[ (X - s) \mathds{1}\{ X \leq s \} \big] \\
		&- X^2 \, F^X(s) - \left( \tfrac{1}{2} (1 - X^2) \cdot s - X \right) d^X(s),\quad s\in\mathbb{R} .
	\end{align*}
\end{corollary}
\begin{remark} \label{Remark for distr under fixed alt}
	Note that for the proof of Theorem $\ref{thm teststat 1 conv under altern}$ we redeployed a line of proof put forward by \cite{BEH:2016}. The asymptotic normality also qualifies our statistics for the applications they propose (see also \cite{BH:2017}).
	
	First, we fix $\alpha \in (0,1)$ and denote by $q_{\alpha} = \Phi^{-1} (1 - \alpha /2)$ the $(1 - \alpha /2)$-quantile of the standard normal distribution. Letting $\widehat{\tau}_{(k),n}^2 = \widehat{\tau}_{(k),n}^2 (X_1, \dots, X_n)$ be a (weakly) consistent estimator of $\tau_{(k)}^2$, $k = 1,2$, figuring in Theorem $\ref{thm teststat 1 conv under altern}$ and Corollary $\ref{coro teststat 2 conv under altern}$, respectively, ($\ref{teststat 1 conv under altern}$) and ($\ref{teststat 2 conv under altern}$) immediately indicate that
	\begin{equation} \label{confidence interval}
		I_n = \left[ \frac{G_n^{(k)}}{n} - \frac{q_{\alpha} \, \widehat{\tau}_{(k),n}}{\sqrt{n}} , \, \frac{G_n^{(k)}}{n} + \frac{q_{\alpha} \, \widehat{\tau}_{(k),n}}{\sqrt{n}} \right]
	\end{equation}
	is an asymptotic confidence interval for $\Delta^{(k)} = \Delta^{(k)}(F)$ at level $1 - \alpha$. Here, $F$ satisfies the assumptions of Theorem  $\ref{thm teststat 1 conv under altern}$ (or Corollary $\ref{coro teststat 2 conv under altern}$). As was briefly explained in the introduction, one main objective of applying Stein's method for the normal distribution is to assess how close a given distribution is to being normal. Thus, seeing $\Delta^{(1)}(F)$ and $\Delta^{(2)}(F)$ as 'measure', in the $L^2$-distance, of how far $F$ differs from the standard normal distribution, we also developed a procedure for empirical assessments of this kind. \\
	Second, we emphasize that our statistics can be employed for inverse testing problems. Namely, if $\Delta_0 > 0$ is a given distance of tolerance, tests that reject $H_{\Delta_0}$ if
	\begin{equation*}
		\frac{G_n^{(k)}}{n} \leq
		\Delta_0 - \frac{\widehat{\tau}_{(k),n}}{\sqrt{n}} \, \Phi^{-1}(1 - \alpha)
	\end{equation*}
	are asymptotic level-$\alpha$-tests for testing
	\begin{equation*}
		H_{\Delta_0} : \Delta^{(k)}(F) \geq \Delta_0 ~~ \text{against} ~~ K_{\Delta_0} : \Delta^{(k)} < \Delta_0 .
	\end{equation*}
	These tests are consistent against each alternative and aim at validating a whole non-parametric neighborhood of the hypothesized, underlying normality. Unfortunately, the direct approach to obtain estimators for $\tau^2_{(k)}$ does not lead to viable results. Trying to replace each expected value and distribution function figuring in the associated covariance kernel by its natural estimator and explicitly calculate the integrals in the definition of $\tau^2_{(k)}$ with those estimators inserted, we observe that some of these integrals can not be solved. In the case of $G_n^{(1)}$ this is due to the occurrence of $p$ in $\mathcal{C}^{(1)}$ which has to be replaced by some kernel density estimator and for $G_n^{(2)}$ the terms involving $\Phi$ and $\omega$ remain intractable for the weight functions we consider (see section $\ref{Section emp results}$). Since these unresolved integrals depend on the data and emerge within double or triple sums, the resulting estimators are impractical. \\
	Finally, we suppose $\{ c_n^{(k)} \} \subset (0, \infty)$ is the sequence of critical values for a level-$\alpha$-test based on $G_n^{(k)}$, $k = 1,2$, i.e., $H_0$ is rejected if $G_n^{(k)} > c_n^{(k)}$ and $\lim_{n \to \, \infty} \mathbb{P}_{H_0}(G_n^{(k)} > c_n^{(k)}) = \alpha$. For an alternative distribution $F$ satisfying the relevant prerequisites of Theorem $\ref{thm teststat 1 conv under altern}$ or Corollary $\ref{coro teststat 2 conv under altern}$, we can approximate the power of the test against this alternative by
	\begin{align} \label{approx power function}
		\mathbb{P}_{F} \big( G_n^{(k)} > c_n^{(k)} \big)
		& = \mathbb{P}_{F} \left( \frac{\sqrt{n}}{\tau_{(k)}} \left\{ \frac{G_n^{(k)}}{n} - \Delta^{(k)} \right\} > \frac{\sqrt{n}}{\tau_{(k)}} \left\{ \frac{c_n^{(k)}}{n} - \Delta^{(k)} \right\} \right) \nonumber \\
		& \stackrel{\sim}{=} 1 - \Phi\left( \frac{\sqrt{n}}{\tau_{(k)}} \left\{ \frac{c_n^{(k)}}{n} - \Delta^{(k)} \right\} \right) .
	\end{align}
	Note that this last application does not need an estimator of $\tau_{(k)}^2$. Instead, $\tau_{(k)}^2$ and $\Delta^{(k)}$ have to be calculated for the particular fixed alternative.
\end{remark}

\section{The Limit null distributions}
\label{limit null}
Our next results for the statistics concern the study of their behavior under the hypothesis. Therefore, we assume, in this section, that $X, X_1, X_2, \dotso \sim \mathcal{N}(0,1)$ are independent random variables. By $\varphi$ we denote the density function of the standard normal law.
\begin{theorem} \label{thm teststat2 H_0 distr}
	There exists a centred Gaussian element $\mathcal{Z}^{(2)}$ of $L^2$ with covariance kernel
	\begin{align*}
		\widetilde{\mathcal{K}}^{(2)} (s,t)
		= \Psi^{(2)}_1(s,t) &+ \Psi^{(2)}_2(s,t) + \Psi^{(2)}_3(s,t) + \Psi^{(2)}_4(s,t) \\
		&+ \Psi^{(2)}_2(t,s) + \Psi^{(2)}_3(t,s) + \Psi^{(2)}_4(t,s) , \quad s,t\in\R,
	\end{align*}
	where
	\begin{align*}
		&\Psi^{(2)}_1(s,t)
		= \varphi^{\prime \prime \prime}(s \wedge t) + (s + t) \varphi^{\prime \prime}(s \wedge t) + (3 + st) \varphi^{\prime}(s \wedge t) + 2(s + t) \varphi (s \wedge t) \\
		& ~~~~~~~~~~~~~~ + (3 + st) \Phi(s \wedge t) + \left( -4 + \tfrac{3}{2} \, st \right) \varphi(s) \varphi(t) - (3 + st) \Phi(s) \Phi(t) , \\
		&\Psi^{(2)}_2(s,t) = 2 \varphi(t) \big[- \varphi^{\prime \prime}(s) - s \varphi^{\prime}(s) \big] , \\
		&\Psi^{(2)}_3(s,t) = \Phi(t) \big[ - \varphi^{\prime \prime \prime}(s) - (s + t) \varphi^{\prime \prime}(s) - (3 + st) \varphi^{\prime}(s) - 2(s + t) \varphi(s) \big] , \\
		&\Psi^{(2)}_4(s,t) = \frac{1}{2} \, t \varphi(t) \big[ \varphi^{\prime \prime \prime}(s) + s \varphi^{\prime \prime}(s) + 2 \varphi^{\prime}(s) \big] ,
	\end{align*}
	with $s \wedge t = \min\{ s, t \}$, such that
	\begin{equation*}
		G_n^{(2)} \stackrel{\mathcal{D}}{\longrightarrow}
		\norm{\mathcal{Z}^{(2)}}_{L^2}^2 ~~ \text{as} ~ n \to \infty .
	\end{equation*}
\end{theorem}
\begin{proof}
	The proof follows the same scheme as the proof of Theorem $\ref{thm teststat 1 conv under altern}$. Starting with ($\ref{teststat 2 after change of var}$) we have
	\begin{equation*}
		G_n^{(2)} = \frac{n}{S_n} \int_{\R} V_n(s)^2 \, \omega \left( \frac{s - \overline{X}_n}{S_n} \right) \mathrm{d}s ,
	\end{equation*}
	where
	\begin{equation*}
		V_n(s) = \widehat{F}_n^X (s) - \Phi\left( \frac{s - \overline{X}_n}{S_n} \right) .
	\end{equation*}
	Using $\sup_{s \, \in \, \R} \left| \varphi^{\prime}(s) \right| < 1$, a Taylor expansion similar to ($\ref{taylor exp F}$) gives
	\begin{equation*}
		\sqrt{n} \, \Phi \left( \frac{s - \overline{X}_n}{S_n} \right)
		\approx \sqrt{n} \left\{ \Phi(s) + \varphi(s) \left( s \cdot \big( \frac{1}{S_n} - 1 \big) - \frac{\overline{X}_n}{S_n} \right) \right\} .
	\end{equation*}
	Plugging the expression on the right hand side, together with approximation ($\ref{approx of F_n^X}$) of $\sqrt{n} \, \widehat{F}_n^X$, into the definition of $V_n$, we get
	\begin{align*}
		\sqrt{n} \, V_n(s)
		\approx \frac{1}{S_n^2} \, \sqrt{n} &\left\{ \frac{1}{n} \sum\limits_{j=1}^{n} X_j (X_j - s) \mathds{1}\{ X_j \leq s \} - \overline{X}_n \, \E \big[ (X - s) \mathds{1}\{X \leq s\} \big] \right. \\
		&\left. ~~\, - S_n^2 \, \Phi(s) - S_n \, \varphi(s) \big( s\cdot (1 - S_n) - \overline{X}_n \big) \vphantom{\sum\limits_{j}^{n}}\right\} .
	\end{align*}
	Since (see the proof of Theorem $\ref{thm teststat 1 conv under altern}$)
	\begin{equation*}
		\sqrt{n} \, S_n^2 \, \Phi(s)
		\approx \frac{1}{\sqrt{n}} \sum\limits_{j=1}^{n} X_j^2 \, \Phi(s)
	\end{equation*}
	and
	\begin{equation*}
		\sqrt{n} \, (1 - S_n) \approx \frac{1}{\sqrt{n}} \sum\limits_{j=1}^{n} \frac{1}{2} \big( 1 - X_j^2 \big) ,
	\end{equation*}
	we obtain
	\begin{equation} \label{asymp repr V_n}
		\sqrt{n} \, V_n(s) \approx \frac{1}{S_n^2} \, \frac{1}{\sqrt{n}} \sum\limits_{j=1}^{n} Z_j(s),
	\end{equation}
	where
	\begin{align*}
		Z_j(s) = & \, X_j (X_j - s) \mathds{1}\{ X_j \leq s \} - X_j \, \E \big[ (X - s) \mathds{1}\{ X \leq s \} \big] \\
		& - X_j^2 \, \Phi(s) - \left( \tfrac{1}{2} \, (1 - X_j^2) \cdot s - X_j \right) \varphi(s) .
	\end{align*}
	Notice that $Z_1, \dots, Z_n$ are iid. random elements of $L^2$ with $\E Z_1 = 0$ ($F^X = \Phi$ under $H_0$). The central limit theorem for separable Hilbert spaces provides the existence of a centred Gaussian element $\mathcal{Z}^{(2)} \in L^2$ with
	\begin{equation*}
		\frac{1}{\sqrt{n}} \sum\limits_{j=1}^{n} Z_j(\cdot)
		\stackrel{\mathcal{D}}{\longrightarrow} \mathcal{Z}^{(2)} (\cdot) .
	\end{equation*}
	By ($\ref{asymp repr V_n}$), $V_n$ is bounded in probability and with ($\ref{unif conv weight function}$), ($\ref{teststat 2 after change of var}$) reads as
	\begin{equation*}
		G_n^{(2)}
		= \norm{\sqrt{n} \, V_n}_{L^2}^2 + o_{\mathbb{P}}(1) .
	\end{equation*}
	Hence, the continuous mapping theorem and Slutzki's Lemma imply
	\begin{equation*}
		G_n^{(2)} \stackrel{\mathcal{D}}{\longrightarrow} \norm{\mathcal{Z}^{(2)}}_{L^2}^2 .
	\end{equation*}
	Since the function $\widetilde{\mathcal{K}}^{(2)}$ defined in the statement of the theorem satisfies $\widetilde{\mathcal{K}}^{(2)}(s,t) = \E[ Z_1(s) \, Z_1(t) ]$ and thus, is the covariance kernel of $\mathcal{Z}^{(2)}$, we are done.
\end{proof}

For $G_n^{(1)}$, the limit distribution under the hypothesis can be obtained in a similar manner. Starting with ($\ref{teststat 1 after change of var}$), we do not need to consider a Taylor expansion, as the term replacing $\Phi$ already has the required 'sum of iid. variables' form. Therefore, the proof is less involved and, following the standard procedure ($\ref{approx of F_n^X}$) for $\hat{F}_n^X$ and applying
\begin{equation*}
	\sqrt{n} \big( 1 - S_n^2 \big) \Phi(s)
	\approx \frac{1}{\sqrt{n}} \sum\limits_{j=1}^{n} \big( 1 - X_j^2 \big) \Phi(s) ,
\end{equation*}
the reasoning closely parallels that of the proof of Theorem $\ref{thm teststat2 H_0 distr}$.
\begin{corollary} \label{coro teststat1 H_0 distr}
	There exists a centred Gaussian element $\mathcal{Z}^{(1)}$ of $L^2$ with covariance kernel
	\begin{equation*}
		\widetilde{\mathcal{K}}^{(1)} (s,t)
		= \Psi^{(1)}_1(s,t) + \Psi^{(1)}_2(s,t) + \Psi^{(1)}_2(t,s) + \Psi^{(1)}_3(s,t) + \Psi^{(1)}_3(t,s) ,\quad s,t\in\R,
	\end{equation*}
	where
	\begin{align*}
		&\Psi^{(1)}_1(s,t)
		= \varphi^{\prime \prime \prime}(s \wedge t) + (s + t) \varphi^{\prime \prime}(s \wedge t) + (1 + st) \varphi^{\prime}(s \wedge t) + (s + t) \varphi (s \wedge t) \\
		& ~~~~~~~~~~~~~~ + (2 + st) \Phi(s \wedge t) - \varphi(s) \varphi(t) - (2 + st) \Phi(s) \Phi(t) , \\
		&\Psi^{(1)}_2(s,t) = \varphi(t) \big[- \varphi^{\prime \prime}(s) - s \varphi^{\prime}(s) \big] , \\
		&\Psi^{(1)}_3(s,t) = \Phi(t) \big[ - \varphi^{\prime \prime \prime}(s) - (s + t) \varphi^{\prime \prime}(s) - (1 + st) \varphi^{\prime}(s) - (s + t) \varphi(s) \big]
	\end{align*}
	such that
	\begin{equation*}
		G_n^{(1)} \stackrel{\mathcal{D}}{\longrightarrow}
		\norm{\mathcal{Z}^{(1)}}_{L^2}^2 ~~ \text{as} ~ n \to \infty .
	\end{equation*}
\end{corollary}
\begin{remark}
	The distribution of $\norm{\mathcal{Z}^{(k)}}_{L^2}^2$, $k = 1, 2$, that is, the limit distribution of $G_n^{(k)}$ under the hypothesis, is that of $\sum_{k=1}^{\infty} \lambda_j^{(k)} N_j^2$. Here, the random variables $N_1, N_2, \dotso \sim \mathcal{N}(0,1)$ are independent, and $\lambda_1^{(k)}, \lambda_2^{(k)}, \dotso$ are the non-zero eigenvalues of the operator
	\begin{equation*}
		L^2 \to L^2, ~~~ f \longmapsto \int_{\R} \widetilde{\mathcal{K}}^{(k)} (\cdot, t) \, f(t) \, \omega(t) \, \mathrm{d}t ,
	\end{equation*}
	$k = 1,2$. Considering the complexity of $\widetilde{\mathcal{K}}^{(k)}$, it does not seem possible to determine $\lambda_j^{(k)}$ (for any weight function) explicitly. Thus, in practice, critical values are obtained by simulation rather then by using asymptotic results. An alternative approach to gain theoretically justified (approximate) critical values is to calculate the first four moments of the limit null distribution and fit a representative of the Pearson- or Johnson- family of distributions to those moments (see \cite{Hen:1990} for the $BHEP$ test). But since we do not face any complications in computing the critical values, we will only pursue the empirical approach.
\end{remark}

\section{Contiguous alternatives}
\label{contiguous alternatives}
Adjusting the argumentation of \cite{HW:1997}, we will derive non-degenerate limit distributions for our statistics under contiguous alternatives converging to the normal distribution at rate $n^{-1/2}$. To this end, we introduce a triangular array of row-wise iid. random variables $X_{n,1}, \dots, X_{n,n}$, $n \in \N$, with Lebesgue density
\begin{equation*}
	p_n (x) = \varphi (x) \cdot \big( 1 + \tfrac{1}{\sqrt{n}} \, c(x) \big),\,x\in\R.
\end{equation*}
Here, $\varphi(x) = \tfrac{1}{\sqrt{2 \pi}} \, e^{-x^2 /2}$, $x\in\R$, and $c : \R \to \R$ is a measurable, bounded function satisfying
\begin{equation*}
	\int_{\R} c(x) \, \varphi(x) \, \mathrm{d}x = 0 .
\end{equation*}
Notice that, by the boundedness of $c$, we may assume $n$ to be large enough to ensure $p_n \geq 0$. First, we set
\begin{equation*}
	\mu_n = \bigotimes\limits_{j=1}^{n} (\varphi \, \mathcal{L}^1) , ~~~
	\nu_n = \bigotimes\limits_{j=1}^{n} (p_n \, \mathcal{L}^1) ,
\end{equation*}
which are measures on $(\R^n, \mathcal{B}^n)$, with $\mathcal{B}^n$ being the Borel-$\sigma$-algebra of $\R^n$. Apparently, $\nu_n$ is absolutely continuous with respect to $\mu_n$ and we can look upon the Radon-Nikodym derivative $L_n = \frac{\mathrm{d}\nu_n}{\mathrm{d}\mu_n}$. We observe that, by a Taylor expansion,
\begin{align*}
	\log \big( L_n (X_{n,1}, \dots, X_{n,n}) \big)
	&= \sum\limits_{j=1}^{n} \log \left( 1 + \frac{1}{\sqrt{n}} \, c(X_{n,j}) \right) \\
	&= \sum\limits_{j=1}^{n} \left( \frac{1}{\sqrt{n}} \, c(X_{n,j}) - \frac{1}{2n} \, c(X_{n,j})^2 \right) + o_{\mathbb{P}}(1) ,
\end{align*}
whenever $(X_{n,1}, \dots, X_{n,n}) \sim \mu_n$. Therefore, viewing $L_n$ as a random element $(\R^n, \mathcal{B}^n, \mu_n) \to (\R, \mathcal{B}^1)$, the Lindeberg-Feller central limit theorem and the (weak) law of large numbers give
\begin{equation*}
	\log \big( L_n \big) \xlongrightarrow{\mathcal{D}_{\mu_n}}
	\mathcal{N} \left( - \frac{\tau^2}{2} , \tau^2 \right) ,
\end{equation*}
where
\begin{equation*}
	\tau^2 = \int_{\R} c(x)^2 \, \varphi(x) \, \mathrm{d}x ~ (< \infty)
\end{equation*}
and $\xlongrightarrow{\mathcal{D}_{\mu_n}}$ denotes convergences in distribution under $\mu_n$. By LeCam's first Lemma (see for instance \cite{HSS:1999}, p.253, Corollary 1), $\nu_n$ is contiguous to $\mu_n$. \\
Now, recall from the proof of Theorem $\ref{thm teststat2 H_0 distr}$
\begin{equation*}
	G_n^{(2)} = \frac{n}{S_n} \int_{\R} V_n(s)^2 \, \omega \left( \frac{s - \overline{X}_n}{S_n} \right) \mathrm{d}s ,
\end{equation*}
where
\begin{equation*}
	V_n(s) = \widehat{F}_n^X (s) - \Phi\left( \frac{s - \overline{X}_n}{S_n} \right) .
\end{equation*}
Also notice that, when interpreting $V_n : \R^n \to L^2$, ($\ref{asymp repr V_n}$) yields
\begin{equation*}
	\norm{\sqrt{n} \, V_n - Z_n^*}_{L^2}^2 = o_{\mu_n}(1) ,
\end{equation*}
where $Z_n^*(s) = \tfrac{1}{\sqrt{n}} \sum_{j=1}^{n} Z_{n,j}(s)$ and
\begin{align*}
	Z_{n,j}(s) = & \, X_{n,j} (X_{n,j} - s) \mathds{1}\{ X_{n,j} \leq s \} - X_{n,j} \, \E \big[ (X_{n,1} - s) \mathds{1}\{ X_{n,1} \leq s \} \big] \\
	& - X_{n,j}^2 \, \Phi(s) - \left( \tfrac{1}{2} \, (1 - X_{n,j}^2) \cdot s - X_{n,j} \right) \varphi(s) .
\end{align*}
Thus, by contiguity,
\begin{equation} \label{asymp equiv under nu_n}
	\norm{\sqrt{n} \, V_n - Z_n^*}_{L^2}^2 = o_{\nu_n}(1) .
\end{equation}
Defining
\begin{equation*}
	\eta (x, s) = x(x - s)\mathds{1}\{ x \leq s \} - \varphi(s) \left( \tfrac{1}{2} (1 - x^2) s - 2x \right) - \Phi(s) (x^2 - xs)
\end{equation*}
and using the boundedness of $c$ we get, under $\mu_n$,
\begin{equation*}
	\lim\limits_{n \, \to \, \infty} \mathrm{Cov} \left( Z_{n,1}(s), c(X_{n,1}) -  \tfrac{1}{2\sqrt{n}} \, c(X_{n,1})^2 \right)
	= \int_{\R} \eta(x,s) \, c(x) \, \varphi(x) \, \mathrm{d}x
	= \zeta(s) ,
\end{equation*}
with $\zeta \in L^2$. Consequently, for any $k \in \N$, $v \in \R^k$ and $s_1, \dots, s_k \in \R$, the multivariate Lindeberg-Feller-type central limit theorem implies,
\begin{align*}
	\frac{1}{\sqrt{n}} \sum\limits_{j=1}^{n} & \left\{
	\left(\begin{array}{c} v_1 \, Z_{n,j}(s_1) + \dots + v_k \, Z_{n,j}(s_k) \\ c(X_{n,j}) - \tfrac{1}{2\sqrt{n}} \, c(X_{n,j})^2 \end{array}\right)
	- \left(\begin{array}{c} 0 \\ - \tfrac{\tau^2}{2 \sqrt{n}} \end{array}\right)
	\right\} \\
	&~~~~~~~~~~~~~~~~~~~~~~~~~ \xlongrightarrow{\mathcal{D}_{\mu_n}}
	\mathcal{N}_{2} \left(
	\boldsymbol{0} ,
	\left(\begin{array}{rr} v^\top \Sigma v \, & \, v^\top \zeta_k \\
		\zeta_k^\top v & \tau^2 \end{array}\right)
	\right).
\end{align*}
Here, $\Sigma = \big( \widetilde{\mathcal{K}}^{(2)}(s_i, s_j) \big)_{1 \leq i,j \leq k}$, with $\widetilde{\mathcal{K}}^{(2)}$ the covariance kernel of $\mathcal{Z}^{(2)}$ from Theorem $\ref{thm teststat2 H_0 distr}$, and $\zeta_k = \big( \zeta(s_1), \dots, \zeta(s_k) \big)^\top$. Therefore,
\begin{equation*}
	\left(\begin{array}{c} v_1 \, Z_n^*(s_1) + \dots + v_k \, Z_n^*(s_k) \\ \log (L_n) \end{array}\right)
	\xlongrightarrow{\mathcal{D}_{\mu_n}}
	\mathcal{N}_{2} \left(
	\left(\begin{array}{c} 0 \\ - \tfrac{\tau^2}{2} \end{array}\right) ,
	\left(\begin{array}{rr} v^\top \Sigma v \, & \, v^\top \zeta_k \\
		\zeta_k^\top v & \tau^2 \end{array}\right)
	\right)
\end{equation*}
and LeCam's third Lemma (see \cite{HSS:1999}, p.259, Lemma 2) implies
\begin{equation} \label{fidi conv}
	Z_n^* \xlongrightarrow{\mathcal{D}_{\nu_n, fidi}}
	\mathcal{Z}^{(2)} + \zeta ,
\end{equation}
where $\xlongrightarrow{\mathcal{D}_{\nu_n, fidi}}$ denotes convergence of the finite-dimensional distributions (under $\nu_n$). We realize that in the proof of Theorem $\ref{thm teststat2 H_0 distr}$ we have shown that $Z_n^* \xlongrightarrow{\mathcal{D}_{\mu_n}} \mathcal{Z}^{(2)}$ which entails the tightness of $Z_n^*$ under $\mu_n$. By contiguity, $Z_n^*$ remains tight under $\nu_n$. Together with ($\ref{fidi conv}$) this demonstrates
\begin{equation*}
	Z_n^* \xlongrightarrow{\mathcal{D}_{\nu_n}}
	\mathcal{Z}^{(2)} + \zeta .
\end{equation*}
Noting that, with ($\ref{unif conv weight function}$), the stochastic boundedness of $Z_n^*$ (under $\nu_n$) and ($\ref{asymp equiv under nu_n}$), $G_n^{(2)}$ reads as
\begin{equation*}
	G_n^{(2)} = \norm{Z_n^*}_{L^2}^2 + o_{\nu_n}(1) .
\end{equation*}
Thus, we have shown
\begin{theorem}
	Under the triangular array $X_{n,1} , \dots, X_{n,n}$, with $X_{n,1} \sim p_n  \mathcal{L}^1$, we have
	\begin{equation*}
		G_n^{(2)} \stackrel{\mathcal{D}}{\longrightarrow}
		\norm{\mathcal{Z}^{(2)} + \zeta}_{L^2}^2 .
	\end{equation*}
\end{theorem}

Since Corollary $\ref{coro teststat1 H_0 distr}$ is obtained with the same line of argument used in Theorem $\ref{thm teststat2 H_0 distr}$ it is evident that we can, likewise, conclude
\begin{equation*}
	G_n^{(1)} \stackrel{\mathcal{D}}{\longrightarrow}
	\norm{\mathcal{Z}^{(1)} + \widetilde{\zeta}}_{L^2}^2,
\end{equation*}
where $\widetilde{\zeta} (s) = \int \widetilde{\eta}(s,x) \, c(x) \, \varphi(x) \, \mathrm{d}x$ and $\widetilde{\eta}(s,x) = (x(x - s) - 1) \mathds{1}\{ x \leq s \} + x \varphi(s) + (1 + sx - x^2) \Phi(s)$. \\
From these statements we discern that tests based on any of our statistics are able to detect contiguous alternatives which converge, at rate $n^{-1/2}$, to the class of normal distributions. For further insights on contiguity, we refer to \cite{Rou:1972} and \cite{Sen:1981}.

\section{Empirical Results}
\label{Section emp results}
For the implementation of our tests, we need to specify the weight function $\omega$. We propose families of test statistics $\left(G_{n,a}^{(k)}: a>0\right)$, $k=1,2$, by considering the density function of a centred normal distribution
\begin{equation*}
	\omega_a (t) = \frac{1}{\sqrt{2 \pi a}} \, e^{- \frac{t^2}{2 a}},\quad t\in\R,
\end{equation*}
where the variance is chosen to be some tuning parameter $a > 0$. Apparently, the continuous function $\omega_a > 0$ satisfies ($\ref{requirements weight function}$) and is, therefore, an admissible weight function in our previous considerations. Note that a weight function of this type has also been employed in \cite{HZ:1990} for the class of $BHEP$ tests. For this explicit function, our statistics have the expressions
\begin{align*}
	G_{n,a}^{(1)}
	= & ~ \frac{2}{n} \sum\limits_{1 \leq j < k \leq n} \left\{ \vphantom{\exp\left(- \tfrac{Y_{(k)}^2}{2 a}\right)} \left( 1 - \Phi\left( \tfrac{Y_{(k)}}{\sqrt{a}} \right) \right) \big( (Y_{(j)}^2 - 1)(Y_{(k)}^2 - 1) + a Y_{(j)} Y_{(k)} \big) \right. \\
	&\left. ~~~~~~~~~~~~~~\, + \frac{a}{\sqrt{2 \pi a}} \, \exp\left(- \tfrac{Y_{(k)}^2}{2 a}\right) \big( - Y_{(j)}^2 Y_{(k)} + Y_{(k)} + Y_{(j)} \big)
	\right\} \\
	&+ \frac{1}{n} \sum\limits_{j=1}^{n} \left\{ \vphantom{\frac{a}{\sqrt{2 \pi a}}} \left( 1 - \Phi\left(\tfrac{Y_{(j)}}{\sqrt{a}}\right) \right) \big( Y_{(j)}^4 + (a - 2) Y_{(j)}^2 + 1 \big) \right. \\
	&\left. ~~~~~~~~~~~ + \frac{a}{\sqrt{2 \pi a}} \, \exp\left(- \tfrac{Y_{(j)}^2}{2 a}\right) \big( 2 Y_{(j)} - Y_{(j)}^3 \big) \right\}
\end{align*}
and
\begin{align*}
	G_{n,a}^{(2)}
	= & ~ \frac{2}{n} \sum\limits_{1 \leq j < k \leq n} \left\{ Y_{(j)} Y_{(k)} \left[ \big( Y_{(j)} Y_{(k)} + a \big) \left( 1 - \Phi\left( \tfrac{Y_{(k)}}{\sqrt{a}} \right) \right) - a Y_{(j)} \, \omega_a(Y_{(k)}) \right] \right\} \\
	&+ \sum\limits_{j=1}^{n} \left\{ \vphantom{\frac{a}{\sqrt{2 \pi (1+a)}} } \frac{Y_{(j)}^2}{n} \left[ (Y_{(j)}^2 + a) \left( 1 - \Phi\left( \tfrac{Y_{(j)}}{\sqrt{a}} \right) \right) - a Y_{(j)} \, \omega_a(Y_{(j)}) \right] \right. \\
	&\left. ~~~~~~~~ - 2 Y_{(j)} \left[ Y_{(j)} \int_{Y_{(j)}}^{\infty} \Phi(t) \, \omega_a (t) \, \mathrm{d}t - a \Phi(Y_{(j)}) \, \omega_a(Y_{(j)}) \right. \right. \\
	& \left. \left. ~~~~~~~~~~~~~~~~\, - \frac{a}{\sqrt{2 \pi (1+a)}} \left( 1 - \Phi\left( \sqrt{\tfrac{1 + a}{a}} \, Y_{(j)} \right) \right) \right] \right\} \\
	&+ n \int_{\R} \Phi(t)^2 \, \omega_a(t) \, \mathrm{d}t ,
\end{align*}
where $Y_{(1)} \leq \dotso \leq Y_{(n)}$ is the ordered sample of the scaled residuals $Y_{n,j}$, see section $\ref{distribution function tests}$. Those expressions make the statistics amenable to computations. We immediately see that the integral figuring in the second sum of $G_{n,a}^{(2)}$, although being accessible to stable numerical integration, is a slight drawback in terms of calculation time as compared to $G_{n,a}^{(1)}$.

The rest of the section is organized as follows. First, we present empirical critical values for a range of possible tuning parameters, sample sizes and significance levels. Then, a brief summary of the competing tests and an overview of the considered alternatives follows. At last, we display the performance of our tests (in comparison to the established tests) in form of a finite-sample power study. Notice that there are several comparative simulation studies for testing normality in the literature, as witnessed by \cite{BDH:1989,FRS:2006,LL:1992,PDB:1977,RDC:2010,SWC:1968,YS:2011} and others. All simulations are performed using the statistical computing environment {\tt R}, see \cite{R:2017}.

%\subsection{Empirical critical values}
%\label{Subsec crit val}
The empirical critical values were obtained by a Monte Carlo simulation with 100~000 repetitions and can be found in Table \ref{emp.val.G1} and Table \ref{emp.val.G2} for sample sizes $n\in\{20,50,100,200,500\}$, tuning parameters $a\in\{0.1,0.25,0.5,1,1.5,2,3\}$ and significance levels $1-\alpha\in\{0.9,0.95,0.99\}$.
\begin{table}[t]
\tiny
	\centering
	\begin{tabular}{c|r|rrrrrrr}
		\hline
		$\alpha$ & $n\backslash a$ & 0.1 & 0.25 & 0.5 & 1 & 1.5 & 2 & 3 \\
		\hline
		\multirow{4}{*}{0.01} & 20 &  1.692 & 1.452 & 1.226 & 0.993 & 0.862 & 0.772 & 0.653 \\
		& 50 & 1.766 & 1.509 & 1.281 & 1.038 & 0.897 & 0.806 & 0.684 \\
		& 100 & 1.784 & 1.538 & 1.303 & 1.052 & 0.914 & 0.820 & 0.694 \\
		& 200 & 1.802 & 1.549 & 1.311 & 1.058 & 0.914 & 0.817 & 0.691 \\
		& 500 & 1.820 & 1.569 & 1.332 & 1.075 & 0.928 & 0.830 & 0.701 \\
		\hline
		\multirow{4}{*}{0.05} &$20$ & 1.027 & 0.891 & 0.754 & 0.609 & 0.525 & 0.469 & 0.396 \\
		& 50 & 1.066 & 0.928 & 0.789 & 0.642 & 0.555 & 0.497 & 0.421 \\
		& 100 & 1.062 & 0.927 & 0.791 & 0.643 & 0.558 & 0.499 & 0.423 \\
		& 200 & 1.067 & 0.929 & 0.795 & 0.649 & 0.563 & 0.504 & 0.427 \\
		& 500 & 1.068 & 0.935 & 0.802 & 0.653 & 0.566 & 0.507 & 0.429 \\
		\hline
		\multirow{4}{*}{0.1} & $20$ &  0.750 & 0.655 & 0.560 & 0.453 & 0.392 & 0.351 & 0.296 \\
		& 50 & 0.763 & 0.670 & 0.578 & 0.472 & 0.411 & 0.368 & 0.312 \\
		& 100 & 0.773 & 0.682 & 0.589 & 0.483 & 0.421 & 0.377 & 0.321 \\
		& 200 & 0.764 & 0.676 & 0.585 & 0.482 & 0.420 & 0.377 & 0.321 \\
		& 500 & 0.770 & 0.678 & 0.588 & 0.484 & 0.422 & 0.379 & 0.323 \\
		\hline
		
	\end{tabular}
	\caption{Empirical $1-\alpha$ quantiles for $G_{n,a}^{(1)}$ (100~000 replications)}\label{emp.val.G1}
\end{table}

\begin{table}[t]
\tiny
	\centering
	\begin{tabular}{c|r|rrrrrrr}
		\hline
		$\alpha$ & $n\backslash a$ & 0.1 & 0.25 & 0.5 & 1 & 1.5 & 2 & 3\\
		\hline
		\multirow{4}{*}{0.01} & 20 &  0.665 & 0.645 & 0.608 & 0.545 & 0.496 & 0.458 & 0.400 \\
		& 50 & 0.711 & 0.691 & 0.652 & 0.582 & 0.530 & 0.489 & 0.430 \\
		& 100 & 0.727 & 0.706 & 0.662 & 0.592 & 0.538 & 0.496 & 0.435 \\
		& 200 & 0.736 & 0.715 & 0.672 & 0.598 & 0.540 & 0.496 & 0.433 \\
		& 500 & 0.754 & 0.731 & 0.687 & 0.607 & 0.549 & 0.506 & 0.440 \\
		\hline
		\multirow{4}{*}{0.05} & 20 & 0.391 & 0.380 & 0.357 & 0.315 & 0.284 & 0.261 & 0.227 \\
		& 50 & 0.420 & 0.408 & 0.384 & 0.342 & 0.309 & 0.284 & 0.247 \\
		& 100 & 0.424 & 0.413 & 0.389 & 0.347 & 0.315 & 0.290 & 0.253 \\
		& 200 & 0.427 & 0.414 & 0.391 & 0.350 & 0.317 & 0.292 & 0.255 \\
		& 500 & 0.433 & 0.419 & 0.396 & 0.354 & 0.320 & 0.295 & 0.258 \\
		\hline
		\multirow{4}{*}{0.1} & 20 & 0.279 & 0.271 & 0.255 & 0.225 & 0.202 & 0.185 & 0.160 \\
		& 50 & 0.293 & 0.287 & 0.271 & 0.241 & 0.219 & 0.201 & 0.175 \\
		& 100 & 0.302 & 0.294 & 0.279 & 0.249 & 0.226 & 0.208 & 0.182 \\
		& 200 & 0.302 & 0.295 & 0.279 & 0.251 & 0.228 & 0.210 & 0.183 \\
		& 500 & 0.303 & 0.295 & 0.280 & 0.251 & 0.228 & 0.210 & 0.184 \\
		\hline
	\end{tabular}
	\caption{Empirical $1-\alpha$ quantiles for $G_{n,a}^{(2)}$ (100~000 replications)}\label{emp.val.G2}
\end{table}

%\subsection{Competing tests and considered alternatives}
%\label{Subsec comp tests}
We considered the following competitors to the new families of test statistics. As classical and well known tests, we included the Shapiro-Wilk test ($SW$), see \cite{SW:1965}, the Shapiro-Francia test ($SF$), see \cite{SF:1972}, and the Anderson-Darling test ($AD$), see \cite{AD:1952}. For the implementation of these tests in {\tt R} we refer to the package {\tt nortest} by \cite{GL:2015}. Tests based on the empirical characteristic function are represented by the Baringhaus-Henze-Epps-Pulley-test ($BHEP$), see \cite{BH:1988,EP:1983}. The $BHEP$ test with tuning parameter $\beta>0$ is based on
\begin{align*}
	BHEP
	&= \frac{1}{n} \sum_{j,k=1}^{n} \exp\left( - \frac{\beta^2}{2} \big(Y_{j}-Y_{k}\big)^2 \right) \\
	& ~~~\, - \frac{2}{\sqrt{1+\beta^2}} \sum_{j=1}^{n} \exp\left( - \frac{\beta^2}{2(1+\beta^2)} \, Y_{j}^2 \right) + \frac{n}{\sqrt{1+2\beta^2}}.
\end{align*}
We fixed $\beta = 1$ and took the critical values from \cite{Hen:1990}, but also restated them in Table \ref{emp.val.Alt}.

Furthermore, we include the quantile correlation test of del Barrio-Cuesta-Albertos-M\'{a}tran-Rodr\'{i}guez-Rodr\'{i}guez (BCMR), based on the $L^2$-Wasserstein distance, see \cite{DBCMRR:1999} and \cite{delBarrio2000} section 3.3. The BCMR statistic is given by
\begin{equation*}
	BCMR =
	n \left( 1 - \frac{1}{S_n^2} \left( \sum_{k=1}^n X_{(k)} \int_{\frac{k-1}{n}}^{\frac{k}{n}} \Phi^{-1}(t) \, \mathrm{d}t \right)^2 \right) - \int_{\frac{1}{n+1}}^{\frac {n}{n+1}} \frac{t (1 - t)}{\left( \varphi\left( \Phi^{-1}(t) \right) \right)^2} \, \mathrm{d}t,
\end{equation*}
where $X_{(k)}$ is the $k$-th order statistic of $X_1,\ldots,X_n$, $S_n^2$ is the sample variance and $\Phi^{-1}$ is the quantile function of the standard normal distribution. Simulated critical values can be found in \cite{K:2009}, or in Table \ref{emp.val.Alt}.

The Henze-Jim\'{e}nez-Gamero test, see \cite{HJG:2018}, uses a weighted $L^2$ distance of the empirical moment generating function $M_n(t) = n^{-1} \sum_{j=1}^n\exp(tY_{n,j})$ and the moment generating function $m(t)$ of the standard normal distribution, $t\in\R$. The test is based on
\begin{align*}
	HJG_\beta
	&= n \int_{\mathbb{R}} \big( M_n(t) -m(t) \big)^2 w_\beta(t) \,\mathrm{d}t \\
	&= \frac1{n\sqrt{\beta}}\sum_{j,k=1}^n \exp\left( \frac{(Y_{n,j} + Y_{n,k})^2}{4 \beta} \right) - \frac{2}{\sqrt{\beta-1/2}}\sum_{j=1}^n\exp\left( \frac{Y_{n,j}^2}{4\beta-2}\right)\\&~~~\,+\frac{n}{\sqrt{\beta-1}}.
\end{align*}
with $\beta > 2$ and $w_\beta(t)=\exp(-\beta t^2)$. We considered the tuning parameters $\beta\in\{2.5,5,10\}$. Since in \cite{HJG:2018} no critical values were simulated in the univariate case, the empirical critical values can be found in Table \ref{emp.val.Alt}. This test was proposed recently, so it is not yet found in any other power studies.

%The Henze-Visagie test, see \cite{HV:2018}, uses a differential equation that characterizes the moment generating function of the standard normal distribution. Denoting the empirical moment generating function $M_n(t) = n^{-1} \sum_{j=1}^n\exp(tY_{j})$ and it's derivative by $M_n^\prime(t)$, the test is based on the weighted $L^2$-statistic
%\begin{align*}
%	HV_\gamma
%	&= n \int_{\mathbb{R}} \big( M_n^\prime(t) - t M_n(t) \big)^2 w_\gamma(t) \,\mathrm{d}t \\
%	&= \frac{\sqrt{\pi}}{n \sqrt{\gamma}} \sum_{j,k=1}^n \exp\left( \frac{(Y_{j} + Y_{k})^2}{4 \gamma} \right) \left( Y_{j} Y_{k} + \big( Y_{j} + Y_{k} \big)^2 \left( \frac{1}{4\gamma^2} - \frac{1}{2 \gamma} \right) + \frac{1}{2 \gamma} \right)
%\end{align*}
%with $\gamma > 2$ and $w_\gamma$ the density function of $\mathcal{N}(0, \gamma)$. Since this test was proposed recently, it is not yet found in power studies. We considered the tuning parameters $\gamma\in\{2.5,5,10\}$.

All of the simulated critical values have been confirmed in a simulation study with 100~000 replications and can, for the sake of completeness, be found in Table \ref{emp.val.Alt} (compare to \cite{Hen:1990,K:2009}).
\begin{table}[t]
\tiny
	\centering
	\begin{tabular}{c|rrrrr}
		\hline
		$n\backslash \mbox{Test}$ & $BCMR$ &  $BHEP$ & $HJG_{2.5}$ & $HJG_{5}$ & $HJG_{10}$\\
		\hline
		20  & 0.31 & 0.368 & 0.1503 & 0.2568  & 0.3258 \\
		50  & 0.30 & 0.374 & 6.349e-3 & 8.568e-3 & 9.556e-3 \\
		100 & 0.29 & 0.376 & 4.001e-4 & 4.988e-4 & 5.369e-4 \\
%		200 & 0.29 & 0.379 &  &  &  \\
		\hline
		
	\end{tabular}
	\caption{Empirical 0.95 quantiles for $BCMR$, $BHEP$ and $HJG_{\beta}$ under $H_0$ (100~000 replications)}\label{emp.val.Alt}
\end{table}

The alternatives were chosen to fit the extensive power study of normality tests by \cite{RDC:2010}, in order to ease the comparison to other tests. Namely, we chose as symmetric distributions the Students $t_\nu$-distribution with $\nu\in \{3,5,10\}$ degrees of freedom, as well as the uniform distribution $\mathcal{U}(-\sqrt{3},\sqrt{3})$. The asymmetric distributions are the $\chi^2_\nu$-distribution with $\nu\in \{5,15\}$ degrees of freedom, the Beta distributions $B(1,4)$ and $B(2,5)$, the Gamma distributions $\Gamma(1,5)$ and $\Gamma(5,1)$, the Gumbel distribution $Gum(1,2)$ with location parameter 1 and scale parameter 2, the lognormal distribution $LN(0,1)$ as well as the Weibull distribution $W(1,0.5)$ with scale parameter 1 and shape parameter 0.5. As representative of bimodal distributions we chose the mixture of normal distributions $Mix\mathcal{N}(p,\mu,\sigma^2)$, where the random variables were generated by
\begin{equation*}
	(1-p)\, \mathcal{N}(0,1) + p \, \mathcal{N}(\mu,\sigma^2), \quad p\in(0,1),\, \mu\in\R, \, \sigma>0.
\end{equation*}

%\subsection{Power studies}
%\label{Subsec power studies}
Since we consider two new families of tests that depend on the choice of the tuning parameter $a$, we will demonstrate the finite-sample power for a range of different parameters. In each Monte Carlo simulation we consider the sample sizes $n=20$, $n=50$ and $n=100$ and fix the nominal level of significance $\alpha$ to $0.05$. Throughout, the critical values for $G_{n,a}^{(k)}$, $k=1,2$, are taken from Table \ref{emp.val.G1} and \ref{emp.val.G2}. Each entry in Table \ref{pow.G} and \ref{pow.A} referring to the finite-sample power of the tests is based on $10~000$ replications. The best performing test for each distribution and sample size taken over both tables has been highlighted for easy reference. Starting from the symmetric distributions, we can see that the $SF$, the $AD$ and $SW$ test have the best performances for these models. Interestingly, the $HJG_\beta$ test performs best for the Students $t_{10}$ distribution, but completely fails to detect the uniform alternative.
The finite-sample power of the new tests for symmetric alternatives is comparable to the $BHEP$ test, which in turn is dominated by the $BCMR$ procedure. Considering the asymmetric distributions, the new procedures show their potential by dominating all other procedures for the $\chi^2$-, Gamma- as well as the Gumbel distribution. Moreover, for small sample sizes the asymmetric Beta model is best detected by $G_{n,a}^{(1)}$ but for larger $n$ the $SW$ and $BCMR$ have small advances in power. In general, $G_{n,a}^{(2)}$ performs a little better than $G_{n,a}^{(1)}$ for all asymmetric models, but in direct comparison the latter procedure seems to be somewhat more robust in detecting symmetric alternatives. All procedures do a good job in rejecting the Weibull and the lognormal alternatives.

%\newpage

\begin{table}[t]
	\tiny
	\setlength{\tabcolsep}{.5mm}
	\centering
	\begin{tabular}{cc|ccccccc|ccccccc}
		& & \multicolumn{7}{c|}{$G_{n,a}^{(1)}$} & \multicolumn{7}{c}{$G_{n,a}^{(2)}$} \\
		Alt. & $n\backslash a$ & 0.1 & 0.25 & 0.5 & 1 & 1.5 & 2 & 3 & 0.1 & 0.25 & 0.5 & 1 & 1.5 & 2 & 3  \\
		\hline
		\multirow{3}{*}{$\mathcal{N}(0,1)$} & 20 & 5 & 5 & 5 & 5 & 5 & 5 & 5 & 5 & 5 & 5 & 5 & 5 & 5 & 5  \\
		& 50 & 5 & 5 & 5 & 5 & 5 & 5 & 5 & 5 & 5 & 5 & 5 & 5 & 5 & 5 \\
		& 100 & 5 & 5 & 5 & 5 & 5 & 5 & 5 & 5 & 5 & 5 & 5 & 5 & 5 & 5 \\
		\hline
		\multirow{3}{*}{$Mix\mathcal{N}(0.3,1,0.25)$} & 20 & 25 & 25 & 24 & 24 & 23 & 23 & 23 & 22 & 22 & 21 & 20 & 20 & 19 & 19 \\
  & 50 & 56 & 57 & 57 & 56 & 56 & 55 & 55 & 50 & 50 & 49 & 48 & 47 & 46 & 46 \\
  & 100 & 86 & 87 & 87 & 87 & 86 & 86 & 86 & 80 & 80 & 79 & 78 & 77 & 77 & 76 \\
		
		\multirow{3}{*}{$Mix\mathcal{N}(0.5,1,4)$} & 20 &  34 & 35 & 35 & 36 & 37 & 37 & 37 & 32 & 32 & 33 & 33 & 33 & 33 & 34 \\
  & 50 & 52 & 56 & 60 & 63 & 64 & 64 & 65 & 42 & 43 & 44 & 46 & 47 & 48 & 49 \\
  & 100 & 75 & 84 & 89 & 91 & 92 & 92 & 92 & 55 & 57 & 60 & 66 & 68 & 70 & 71 \\
		
		\hline
		\multirow{3}{*}{$t_3$} & 20 &  30 & 31 & 33 & 34 & 35 & 35 & 35 & 32 & 32 & 33 & 34 & 34 & 35 & 35 \\
  & 50 & 41 & 46 & 50 & 54 & 56 & 57 & 58 & 44 & 45 & 47 & 50 & 52 & 53 & 54 \\
  & 100 & 54 & 63 & 70 & 76 & 78 & 79 & 81 & 53 & 56 & 61 & 67 & 70 & 72 & 74 \\
		
		\multirow{3}{*}{$t_5$} & 20 &  16 & 17 & 18 & 19 & 19 & 19 & 20 & 18 & 18 & 19 & 19 & 19 & 19 & 20 \\
  & 50 & 22 & 25 & 27 & 29 & 31 & 31 & 32 & 26 & 26 & 27 & 29 & 30 & 31 & 31 \\
  & 100 & 27 & 31 & 36 & 41 & 43 & 44 & 46 & 30 & 31 & 34 & 37 & 39 & 41 & 42 \\
		\multirow{3}{*}{$t_{10}$} & 20 &   9 & 9 & 10 & 10 & 10 & 10 & 10 & 10 & 10 & 10 & 10 & 11 & 11 & 11 \\
  & 50 & 11 & 11 & 12 & 13 & 14 & 15 & 15 & 13 & 13 & 13 & 14 & 15 & 15 & 16 \\
  & 100 & 11 & 12 & 14 & 16 & 17 & 17 & 18 & 14 & 14 & 15 & 16 & 17 & 17 & 18 \\
		
		\hline
		\multirow{3}{*}{$\mathcal{U}(-\sqrt{3},\sqrt{3})$} & 20 & 4 & 4 & 3 & 3 & 3 & 3 & 3 & 2 & 2 & 2 & 2 & 1 & 1 & 1 \\
  & 50 & 3 & 4 & 5 & 7 & 7 & 8 & 8 & 2 & 2 & 2 & 2 & 1 & 1 & 1 \\
  & 100 & 5 & 10 & 19 & 32 & 38 & 41 & 45 & 2 & 2 & 3 & 3 & 3 & 3 & 3 \\
		\hline
		\multirow{3}{*}{$\chi^2_5$} & 20 & 44 & {\bf 45} & {\bf 45} & {\bf 45} & 44 & 44 & 44 & 45 & 45 & 44 & 44 & 43 & 43 & 43 \\
  & 50 & 84 & 86 & 87 & 87 & 87 & 87 & 87 & 87 & 87 & 87 & 87 & 87 & 87 & 87 \\
  & 100 & 99 & 99 & 99 & 99 & 99 & 99 & 99 & 99 & 99 & 99 & 99 & 99 & 99 & 99 \\
		
		\multirow{3}{*}{$\chi^2_{15}$} & 20  & 18 & 18 & {\bf 19} & {\bf 19} & {\bf 19} & {\bf 19} & {\bf 19} & {\bf 19} & {\bf 19} & {\bf 19} & {\bf 19} & {\bf 19} & {\bf 19} & {\bf 19} \\
  & 50 & 44 & 45 & 46 & 46 & 46 & 46 & 46 & {\bf 47} &{\bf  47} & {\bf 47} & {\bf 47} & 46 & 46 & 46 \\
  & 100 & 74 & 76 & 77 & 78 & 78 & 78 & 78 & 78 & 78 &{\bf  79} & 78 & 78 & 78 & 78 \\

		\hline
		
		\multirow{3}{*}{$B(1,4)$} & 20 &  49 & 51 & 51 & 51 & 51 & 50 & 50 & 49 & 49 & 49 & 48 & 47 & 47 & 46 \\
  & 50 & 90 & 93 & 94 & 94 & 94 & 94 & 94 & 92 & 92 & 92 & 92 & 92 & 92 & 92 \\
  & 100 & {\bf 100} & {\bf 100} & {\bf 100} & {\bf 100} & {\bf 100} & {\bf 100} & {\bf 100} & {\bf 100} & {\bf 100} & {\bf 100} & {\bf 100} & {\bf 100} & {\bf 100} & {\bf 100} \\

		\multirow{3}{*}{$B(2,5)$} & 20 &  15 & {\bf 16} & 15 & 15 & 15 & 15 & 15 & 15 & 15 & 15 & 15 & 14 & 14 & 14 \\
  & 50 & 40 & 42 & 43 & 44 & 44 & 43 & 43 & 42 & 42 & 42 & 42 & 41 & 41 & 41 \\
  & 100 & 73 & 77 & 79 & 80 & 81 & 81 & 81 & 76 & 77 & 78 & 78 & 78 & 78 & 78 \\
		
		\hline
		
		\multirow{3}{*}{$\Gamma(1,5)$} & 20 &  76 & 78 & 78 & 78 & 78 & 78 & 78 & 77 & 77 & 76 & 76 & 75 & 75 & 74 \\
  & 50 & 99 & {\bf 100} & {\bf 100} & {\bf 100} & {\bf 100} & {\bf 100} & {\bf 100} & 99 & 99 & {\bf 100} & {\bf 100} & {\bf 100} & {\bf 100} & 99 \\
  & 100 & {\bf 100} & {\bf 100} & {\bf 100} & {\bf 100} & {\bf 100} & {\bf 100} & {\bf 100} & {\bf 100} & {\bf 100} & {\bf 100} & {\bf 100} & {\bf 100} & {\bf 100} & {\bf 100} \\

		\multirow{3}{*}{$\Gamma(5,1)$} & 20 & 25 & 25 & 25 & 25 & 25 & 25 & 25 & 26 & 26 & 26 & 25 & 25 & 25 & 25 \\
  & 50 & 58 & 60 & 61 & 61 & 61 & 61 & 61 & {\bf 62} & {\bf 62} & {\bf 62} & {\bf 62} & {\bf 62} & 61 & 61 \\
  & 100 & 88 & 90 & 90 & {\bf 91} & {\bf 91} & {\bf 91} & {\bf 91} & {\bf 91} & {\bf 91} & {\bf 91} & {\bf 91} & {\bf 91} & {\bf 91} & {\bf 91} \\

		\hline
		
		\multirow{3}{*}{$W(1,0.5)$} &20 & 76 & 78 & 78 & 78 & 78 & 78 & 78 & 77 & 77 & 77 & 76 & 75 & 75 & 75 \\
  & 50 & 99 & {\bf 100} & {\bf 100} & {\bf 100} & {\bf 100} & {\bf 100} & {\bf 100} & 99 & 99 & {\bf 100} & {\bf 100} & 99 & 99 & 99 \\
  & 100 & {\bf 100} & {\bf 100} & {\bf 100} & {\bf 100} & {\bf 100} & {\bf 100} & {\bf 100} & {\bf 100} & {\bf 100} & {\bf 100} & {\bf 100} & {\bf 100} & {\bf 100} & {\bf 100} \\

		\multirow{3}{*}{$Gum(1,2)$} & 20 &  32 & 33 & 33 & 33 & 33 & 33 & 33 & {\bf 34} & {\bf 34} & {\bf 34} & 33 & 33 & 33 & 33 \\
  & 50 & 70 & 71 & 72 & 72 & 72 & 72 & 72 & {\bf 73} & {\bf 73} & {\bf 73} & 72 & 72 & 72 & 72 \\
  & 100 & 94 & 95 & 95 & 95 & {\bf 96} & {\bf 96} & {\bf 96} & {\bf 96} & {\bf 96} & {\bf 96} & {\bf 96} & {\bf 96} & {\bf 96} & {\bf 96} \\

		\multirow{3}{*}{$LN(0,1)$} & 20 &  90 & 91 & 91 & 91 & 91 & 91 & 91 & 90 & 90 & 90 & 90 & 90 & 89 & 89 \\
  & 50 & {\bf 100} & {\bf 100} & {\bf 100} & {\bf 100} & {\bf 100} & {\bf 100} & {\bf 100} & {\bf 100} & {\bf 100} & {\bf 100} & {\bf 100} & {\bf 100} & {\bf 100} & {\bf 100} \\
  & 100 & {\bf 100} & {\bf 100} & {\bf 100} & {\bf 100} & {\bf 100} & {\bf 100} & {\bf 100} & {\bf 100} & {\bf 100} & {\bf 100} & {\bf 100} & {\bf 100} & {\bf 100} & {\bf 100} \\

	\end{tabular}
	\caption{Empirical rejection rates for $G_{n,a}^{(j)}$, $j=1,2$ ($\alpha=0.05$, 10~000 replications)}\label{pow.G}
\end{table}

\begin{table}[t]
\tiny
	\setlength{\tabcolsep}{.5mm}
	\centering
	\begin{tabular}{cc|cccccccc}

		Alt. & $n$ & $SW$ & $BCMR$ & $BHEP$ & $AD$ & $SF$ & $HJG_{2.5}$ & $HJG_{5}$ & $HJG_{10}$\\
		\hline
		\multirow{3}{*}{$\mathcal{N}(0,1)$} & 20 & 5 & 5 & 5 & 5 & 5 & 5 & 5 & 5 \\
		& 50 & 5 & 5 & 5 & 5 & 5 & 5 & 5 & 5  \\
		& 100 & 5 & 5 & 5 & 5 & 5 & 5 & 5 & 5 \\
		\hline
		\multirow{3}{*}{$Mix\mathcal{N}(0.3,1,0.25)$} & 20 & 28 & 28 & 27 & {\bf 30} & 25 & 11 & 13 & 14 \\
  & 50 & 60 & 60 & 62 & {\bf 68} & 57 & 16 & 26 & 32 \\
  & 100 & 89 & 89 & 90 & {\bf 94} & 88 & 28 & 49 & 58 \\
		\multirow{3}{*}{$Mix\mathcal{N}(0.5,1,4)$} & 20 & 40 & 43 & 42 & 46 & {\bf 48} & 34 & 33 & 33 \\
  & 50 & 78 & 80 & 80 & {\bf 86} & 83 & 49 & 49 & 46 \\
  & 100 & 97 & 98 & 98 & {\bf 99} & 98 & 69 & 68 & 61 \\
		\hline
		\multirow{3}{*}{$t_3$} & 20 &  35 & 37 & 34 & 33 & {\bf 40} & 38 & 37 & 36 \\
  & 50 & 64 & 65 & 61 & 60 & {\bf 69} & 64 & 62 & 59 \\
  & 100 & 88 & 89 & 86 & 85 & {\bf 91} & 86 & 84 & 78 \\
		\multirow{3}{*}{$t_5$} & 20 & 19 & 20 & 18 & 17 & {\bf 22} & {\bf 22} & {\bf 22} & 21 \\
  & 50 & 35 & 37 & 32 & 31 & {\bf 41} & 40 & 38 & 36 \\
  & 100 & 56 & 58 & 50 & 48 & {\bf 63} & 59 & 55 & 50 \\
		\multirow{3}{*}{$t_{10}$} & 20 & 10 & 11 & 9 & 9 & {\bf 12} & {\bf 12} & {\bf 12} & {\bf 12} \\
  & 50 & 16 & 17 & 13 & 12 & {\bf 20} & {\bf 20} & 19 & 18 \\
  & 100 & 22 & 24 & 16 & 15 & {\bf 28} & {\bf 28} & 26 & 23 \\
		\hline
		\multirow{3}{*}{$\mathcal{U}(-\sqrt{3},\sqrt{3})$} & 20 &  {\bf 21} & 17 & 13 & 17 & 8 & 0 & 0 & 0 \\
  & 50 & {\bf 75} & 70 & 55 & 58 & 47 & 0 & 0 & 0 \\
  & 100 & {\bf 100} & 99 & 95 & 95 & 97 & 0 & 0 & 0 \\
		\hline
		\multirow{3}{*}{$\chi^2_5$} & 20 & 44 & 44 & 42 & 38 & 42 & 33 & 36 & 39 \\
  & 50 & {\bf 89} & 88 & 84 & 80 & 85 & 65 & 76 & 80 \\
  & 100 & {\bf 100} & {\bf 100} & 99 & 99 & {\bf 100} & 91 & 98 & 99 \\

		\multirow{3}{*}{$\chi^2_{15}$} & 20 & 18 & 18 & 17 & 16 & 18 & 16 & 17 & 18 \\
  & 50 & 42 & 42 & 39 & 33 & 40 & 32 & 38 & 41 \\
  & 100 & 75 & 74 & 68 & 61 & 71 & 54 & 68 & 73 \\
		\hline
	
		\multirow{3}{*}{$B(1,4)$} & 20 & {\bf 59} & 58 & 52 & 51 & 53 & 28 & 34 & 38 \\
  & 50 & {\bf 98} & {\bf 98} & 94 & 95 & 97 & 57 & 76 & 83 \\
  & 100 & {\bf 100} & {\bf 100} & {\bf 100} & {\bf 100} & {\bf 100} & 89 & 99 & {\bf 100} \\
		\multirow{3}{*}{$B(2,5)$} & 20 & {\bf 16} & {\bf 16} & {\bf 16} & 14 & 14 & 9 & 11 & 12 \\
  & 50 & 50 & 47 & 45 & 39 & 40 & 16 & 25 & 30 \\
  & 100 & 90 & 89 & 80 & 76 & 82 & 29 & 54 & 64 \\
		\hline
		\multirow{3}{*}{$\Gamma(1,5)$} & 20 &  {\bf 83} & {\bf 83} & 77 & 77 & 80 & 57 & 63 & 67 \\
  &50 & {\bf 100} & {\bf 100} & {\bf 100} & {\bf 100} & {\bf 100} & 91 & 97 & 98 \\
  & 100 & {\bf 100} & {\bf 100} & {\bf 100} & {\bf 100} & {\bf 100} & {\bf 100} & {\bf 100} & {\bf 100} \\
		\multirow{3}{*}{$\Gamma(5,1)$} & 20 & 24 & 24 & 23 & 20 & 24 & 20 & 22 & 23 \\
  & 50 & 59 & 59 & 55 & 49 & 56 & 42 & 51 & 55 \\
  & 100 & 90 & 90 & 85 & 81 & 88 & 69 & 83 & 87 \\
		\hline
		\multirow{3}{*}{$W(1,0.5)$} & 20 &  {\bf 84} & 83 & 78 & 77 & 80 & 58 & 64 & 67 \\
  & 50 & {\bf 100} & {\bf 100} & {\bf 100} & {\bf 100} & {\bf 100} & 91 & 97 & 98 \\
  & 100 & {\bf 100} & {\bf 100} & {\bf 100} & {\bf 100} & {\bf 100} & {\bf 100} & {\bf 100} & {\bf 100} \\
		
		\multirow{3}{*}{$Gum(1,2)$} & 20 & 31 & 32 & 31 & 28 & 32 & 28 & 30 & 31 \\
  & 50 & 69 & 69 & 66 & 60 & 67 & 55 & 64 & 67 \\
  & 100 & 94 & 94 & 91 & 89 & 93 & 83 & 92 & 94 \\
		
		\multirow{3}{*}{$LN(0,1)$} & 20 &  {\bf 93} & {\bf 93} & 91 & 90 & 91 & 78 & 83 & 85 \\
  & 50 & {\bf 100} & {\bf 100} & {\bf 100} & {\bf 100} & {\bf 100} & 99 & {\bf 100} & {\bf 100} \\
  & 100 & {\bf 100} & {\bf 100} & {\bf 100} & {\bf 100} & {\bf 100} & {\bf 100} & {\bf 100} & {\bf 100} \\
	\end{tabular}
	\caption{Empirical rejection rates for competing procedures ($\alpha=0.05$, 10~000 replications)}\label{pow.A}
\end{table}

\section{Conclusions and outlines}
\label{Section conclusions}
Starting with Charles Stein's insight that a random variable $X$ has a unit normal distribution if, and only if,
\begin{equation*}
	\E \big[ f^{\prime}(X) - Xf(X) \big] = 0
\end{equation*}
holds for any absolutely continuous function, we developed two classes of goodness-of-fit statistics for testing the normality hypothesis. To that end, we utilized the zero-bias transformation to bypass the problem of calculating an empirical property for all absolutely continuous function. An advantage of the underlying zero-bias transformation over many other types of transformation applied in goodness-of-fit testing, like the characteristic function or the Laplace transform, is that the distribution inserted into the mapping is not associated with a purely analytic quantity but is mapped to another distribution and, thereby, stays accessible to a stochastically intuitive examination (cf. Proposition $\ref{Prop existence zero bias}$ and Lemma $\ref{prelim lemma}$). The conducted power study suggests that our tests are serious competitors to established tests and, in a noteworthy number of asymmetrical alternatives, even set new markers in terms of the highest power achieved. Additionally, the statistics possess the most important asymptotic properties desirable for a hypothesis test, namely, they are consistent against general alternatives, have a limiting normal distribution under fixed alternatives and are able to detect contiguous alternatives.
We want to emphasize that some questions remain open for further research. An interesting question is whether there are some limiting statistics as $a \to \infty$ for the families of statistics $\left(G_{n,a}^{(j)}\right)$, $j=1,2$, with weight function $\omega_a$. In view of the simulation results of section \ref{Section emp results}, finding a (possibly data dependent) best choice of tuning parameter $a$ would be a nice result, but for most models the power performance seems to be rather stable. An exception are the symmetric distributions where the influence seems to be significant, cf. the results for the uniform distribution. Perhaps a data-dependent choice of the tuning parameter as proposed by \cite{AS:2015} can give better results (for a discussion of this method for exponentiality tests with tuning parameters, see section 3 of \cite{ASSV:2017}).   Due to the problems with the classical, direct approach (see the Remark in section $\ref{Section under fixed alternatives}$) we have not stated consistent estimators for $\tau_{(1)}^2$ and $\tau_{(2)}^2$, which are crucial to find asymptotic confidence intervals as in (\ref{confidence interval}).

\bibliography{lit-BibTexNV_neu}   % name your BibTeX data base
\bibliographystyle{abbrv}      % mathematics and physical sciences
S. Betsch and B. Ebner, Institute of Stochastics, Karlsruhe Institute of Technology (KIT), Englerstr. 2, D-76133 Karlsruhe:
\\
{\texttt Steffen.Betsch@student.kit.edu} \ \ {\texttt Bruno.Ebner@kit.edu}

\end{document}